\documentclass[article, 10pt, twocolumn]{IEEEtran}
\usepackage{graphicx}
\usepackage{float}
\usepackage{amsfonts}
\usepackage{amsmath}
\usepackage{amssymb}
\usepackage{cite}
\usepackage{bm}
\usepackage{mathrsfs}
\usepackage{epsfig}
\usepackage{algorithm}
\usepackage{algorithmic}
\usepackage{amsthm}
\usepackage{color}
\usepackage{graphicx}
\usepackage{caption}
\usepackage{subcaption}
\usepackage[utf8]{inputenc}

\newtheorem{proposition}{Proposition}

\newtheorem{note}{Note}

\begin{document}
\title{Sparse Signal Detection with Compressive Measurements via Partial Support Set Estimation}

\author{\authorblockA{Thakshila Wimalajeewa  \emph{Member, IEEE}, and Pramod K.
Varshney, \emph{Fellow IEEE}\\
Department of EECS, Syracuse University, Syracuse, NY.
Email: \{twwewelw, varshney\}@syr.edu}
}

\maketitle\thispagestyle{empty}

\begin{abstract}
In this paper, we consider the problem of  sparse signal detection   based on   partial support set estimation with compressive measurements in a distributed network. Multiple  nodes in the network are assumed to observe  sparse signals which share  a common but unknown support.    While in the traditional  compressive sensing (CS) framework, the goal is to recover the complete  sparse signal, in sparse signal detection, complete signal recovery may not be necessary to make a reliable detection decision. In particular,  detection can be performed based on  partially or inaccurately estimated signals which requires less computational burden than that is required for complete signal recovery. To that end, we investigate the problem of sparse signal detection based on  partially  estimated   support set. First, we discuss how to determine  the minimum  fraction of the support set  to be known  so that a desired detection performance  is achieved in a centralized  setting. Second, we develop two distributed algorithms for sparse signal detection when the raw compressed observations are not available at the central fusion center. In these algorithms, the final decision  statistic is computed based on locally estimated partial support sets   via orthogonal matching pursuit (OMP) at individual nodes. The proposed distributed algorithms with  less communication overhead are shown to provide comparable  performance (sometimes better) to the centralized approach when the size of the estimated partial support set is very small.
\end{abstract}

\begin{IEEEkeywords}
Compressive sensing, sparse signal detection, partial support set estimation, orthogonal matching pursuit (OMP)
\end{IEEEkeywords}

\footnotetext[1]{This work is partly  supported by the National Science
Foundation (NSF) under Grant No. 1307775.}

\IEEEpeerreviewmaketitle

\section{Introduction}
Sparsity is one of the low dimensional structures exhibited in many signals of interest including audio, video and radar signals. A signal is said to be sparse if the coefficient vector, when represented in a known (orthogonal) basis,  contains only a few significant elements. Sparsity has  been exploited in signal processing and approximation theory
for tasks such as compression, denoising, model selection, and image processing  \cite{Davenport_book2012} for a long time. The problem of sparse signal recovery has attracted much attention in the recent literature with advancements of the theory of compressive sensing (CS). In the CS framework,  a sparse signal can be reliably recovered with a small number of random projections under certain conditions  \cite{candes1,candes2,donoho1,candes_TIT1,Eldar_B1,Baraniuk4}.

In addition to complete recovery, the problem of detecting signals which are sparse   is  important  in many applications including sensor,   cognitive radio, and radar  networks. For this problem, complete signal recovery is not necessary to make a reliable detection decision. Theories and concepts developed in CS for sparse signal recovery have  been exploited in the recent literature for signal detection problems  \cite{duarte_ICASSP06,
haupt_ICASSP07,Wang_ICASSP08,meng_CISS09,davenport_JSTSP10,Wimalajeewa_asilomar10,
Zahedi_Phys12,Wimalajeewa_ICASSP13,Gang_globalsip14,Bhavya_cscps14,Bhavya_asilomar14, Zheng_icc11,Rao_icassp2012,Cao_Info2014,Kailkhura_WCL16,Kailkhura_TSP16}. These works include deriving  performance bounds for CS based signal detection \cite{haupt_ICASSP07,Wang_ICASSP08,davenport_JSTSP10,
Wimalajeewa_asilomar10,Bhavya_cscps14,Bhavya_asilomar14,Rao_icassp2012,Cao_Info2014,Kailkhura_TSP16}, developing algorithms \cite{duarte_ICASSP06,Wimalajeewa_ICASSP13,Gang_globalsip14} and designing low dimensional projection matrices  \cite{Zahedi_Phys12,Kailkhura_WCL16}.
 While some of the work, such as \cite{duarte_ICASSP06,haupt_ICASSP07,Wimalajeewa_ICASSP13,Gang_globalsip14,Zheng_icc11,Rao_icassp2012,Cao_Info2014} focused  on sparse signal detection when  the underlying subspace where the  signal lies is unknown, some other works \cite{davenport_JSTSP10,Wimalajeewa_asilomar10,Zahedi_Phys12,Bhavya_cscps14,Bhavya_asilomar14} considered  the problem of detecting  signals which are not necessarily sparse.
 In CS based  sparse signal detection, the main objective is to utilize a small   number of measurements  to extract decision  statistics  to make a reliable detection decision.  For example,  in  \cite{haupt_ICASSP07}, the authors consider the case where an estimate of the sparse signal is obtained using some additional information to implement the matched filter for detection. In \cite{duarte_ICASSP06}, the authors consider a greedy algorithm developed for sparse signal reconstruction  where the detection decision is made after the maximum absolute  coefficient of the partially  estimated sparse signal exceeds a certain threshold. In this case, if the maximum component  of the estimated signal exceeds the required threshold during  an early iteration, complete signal recovery is not  necessary.

Without completely reconstructing the signal, one possible way to perform reliable detection is to construct a decision statistic by estimating  only a fraction of the support  set of a sparse signal.  Using   greedy approaches  as considered in \cite{duarte_ICASSP06},  a  partial support set can be   computed with  fewer  number of iterations (thus less computational power and time),  than that is required for complete signal recovery. Then, a decision statistic  is  computed based on the corresponding nonzero coefficients obtained via least squares  estimation. When the requirement is to estimate only a fraction of the support,  different approaches developed for partial support set estimation can be used with fewer  compressive measurements  \cite{Reeves_ISIT08}. On the other  hand, with multiple sensors each observing sparse signals with joint sparse structures, the   detection performance can be enhanced by proper fusion of partially estimated support sets at individual nodes. Motivated by these, in this paper, we investigate  the problem of sparse signal detection based on partially estimated support sets  with multiple sensors in centralized as well as distributed settings.

We consider a distributed network in which the sparse signals observed by multiple nodes share the same sparsity pattern.  Each node makes CS based measurements. While sparse signal (or the complete support of  the sparse signal) recovery with compressed measurements  has been investigated quite extensively with  multiple sensors with a common sparse support model \cite{Tropp_P12006,Tropp_P22006,Cotter1,Chen2,Obozinski2,Wipf2,Eldar4,Eldar1,
ling1,Zeng1,Bazerque1,Ling2,Wimalajeewa_ICASSP13,Rabbat1,Patterson1},  the problem of sparse signal detection  with CS based measurements  in a multiple sensor setting has not been investigated adequately. Recently, our work presented in \cite{Wimalajeewa_ICASSP13,Gang_globalsip14},   extended the decentralized algorithms developed for joint sparsity pattern recovery  based on  greedy techniques to the case of  sparse signal detection in a decentralized manner. In \cite{Wimalajeewa_ICASSP13}, a heuristic decision statistic based on support set indices estimated at multiple nodes via orthogonal matching pursuit (OMP) was computed in a decentralized manner. This  approach counts the number of nodes estimating the same index at a given iteration of OMP, therefore,  it is promising only when the network has sufficient number of nodes or the signal-to-noise-ratio (SNR) is relatively large.  In \cite{Gang_globalsip14}, a  similar approach  has been considered using the subspace pursuit (SP) algorithm.    The  current work is  different from \cite{Wimalajeewa_ICASSP13,Gang_globalsip14}  in terms of the  computation of decision statistics and the communication architectures  used.

In this paper, we consider the problem in which  the  detection decision is made based on the knowledge of a  partial support set.  First, we discuss how to obtain  the minimum fraction of the support set to make a detection decision with a desired performance level in a centralized setting. Second, we extend the OMP algorithm in  centralized and distributed settings  to perform joint partial support set  estimation and sparse signal detection. Note that,  in the centralized setting all the compressed observations are transmitted to the fusion center. In the distributed setting, two approaches are considered in which local decision statistics are computed based on  partial support set estimates  at the individual nodes. In the  first distributed algorithm, based on independently estimated partial support sets, a local decision statistic is computed and transmitted to the fusion center. In the second approach, independently estimated partial support sets are fused at the fusion center  to obtain an  updated support set of larger size which are fed back to the nodes. In this case, it is possible to obtain  the  complete support set at the fusion center   under certain conditions.  Then, a decision statistic computed based on the updated support set is transmitted to the fusion center    to compute  the final decision statistic. The two distributed algorithms differ from each other  in terms of  the communication overhead required. Another interesting observation is that, under certain conditions, the two distributed algorithms (with less communication burden) perform better  than the centralized algorithm which requires a higher communication burden.

\subsection*{Organization of the paper}
The paper is organized as follows. In Section \ref{sec_problem}, the sparse signal detection problem with compressive measurements in a distributed network  is formulated. The minimum fraction of the support set required to be known in order to achieve a desired detection  performance in a centralized setting is derived  in Section \ref{section_knownsupport}.  Several practical algorithms based on OMP to perform sparse signal detection by partial support set estimation in centralized as well as in distributed settings   are proposed  in Section \ref{sec_OMP_detection}. In Section \ref{sec_simulation}, numerical results are presented to show the effectiveness of the proposed algorithms. Finally, concluding remarks are given in Section \ref{sec_conclusion}.

\subsection*{Notation}
The following notation and terminology are  used throughout the paper.
Scalars are denoted by lower case letters and symbols; e.g., $x$ and $\alpha$. Lower case boldface letters are used to denote vectors; e.g., $\mathbf x$. Both upper case boldface letters and boldface symbols  are used to denote matrices, e.g.,  $\mathbf A$, $\boldsymbol\Phi$.  Matrix transpose is denoted by $\mathbf A^T$. The $l_p$ norm of a vector $\mathbf x$ is denoted by $||\mathbf x||_p$. Calligraphic letters are used to denote sets; e.g., $\mathcal U$. The notation $\mathcal U \setminus \mathcal V$ represents the set of elements in $\mathcal U$ which  are not in $\mathcal V$ when $\mathcal V \subseteq \mathcal U$.  By $\mathbf B(\mathcal U)$, we denote the submatrix of $\mathbf B$ with columns indexed in $\mathcal U$.  We use the notation $|.|$ to denote the absolute value of a scalar, as well as the cardinality of a set. We use $\mathbf I_N$ to denote the identity matrix of dimension $N$ (we avoid using subscript when there is no ambiguity) and $\mathbf 0$ to denote the vector  of all zeros with an appropriate dimension. The notation $\mathbf x \sim \mathcal N(\boldsymbol\mu, \boldsymbol\Sigma)$ denotes that the random vector  $\mathbf x$ is multivariate Gaussian with mean vector $\boldsymbol\mu$ and covariance matrix $\boldsymbol\Sigma$. The notation $x\sim \mathcal X_k^2$ denotes that the random variable $x$ is distributed as a chi-squared with $k$ degrees of freedom while  $x\sim \mathcal X_k^2(\lambda)$ denotes that $x$ has a non-central chi-squared distribution with non-centrality parameter $\lambda$.

\section{Sparse Signal Detection  with Multiple Sensors}\label{sec_problem}
\subsection{Observation model with uncompressed data}
Consider the problem of  detecting  unknown (deterministic)  sparse signals   in the presence of noise based on observations collected at multiple sensor nodes. Let there be $L$ nodes.  The observation model at the $j$-th node  under hypothesis   $\mathcal H_1$, (the signal is present) and  $\mathcal H_0$ (the signal is absent) is given by
\begin{eqnarray}
\mathcal H_1: ~ \mathbf x_j &=& \boldsymbol\theta_j + \mathbf v_j \nonumber\\
\mathcal H_0: ~ \mathbf x_j &=& \mathbf v_j\label{obs_uncomp}
\end{eqnarray}
where $\boldsymbol\theta_j$ is the signal observed  by the $j$-th node and $\mathbf v_j$ is the additive noise for $j=1, \cdots, L$.  The  signal $\boldsymbol\theta_j$,  is assumed to be  sparse in an orthonormal basis  $\boldsymbol\Psi_j$ so that  $\boldsymbol\theta_j = \boldsymbol\Psi_j \mathbf s_j$ where $\mathbf s_j$  is the  $N\times 1$ coefficient vector with only $k\ll N$ nonzero elements. We assume that all the signals are sparse in the same  basis so that $\boldsymbol\Psi_j = \boldsymbol\Psi$ and the coefficient vectors $\mathbf s_j$'s share the same sparse support. This particular joint sparse model  is applicable in sensor networks where multiple sensors  observe a signal that is sparse in the same basis, and the amplitudes of coefficients are different from each other due to different propagation conditions \cite{Baron1}.   We further assume that the sparsity index $k$ is known in advance. Finding the  sparsity index of a sparse signal and estimating the complete sparse signal (or the  support) are conceptually different and there are several algorithms developed to achieve the former \cite{Lopes_2013,Bioglio_2015}.    The noise vector $\mathbf v_j$ is  assumed to be Gaussian with $\mathbf v_j \sim \mathcal N(\mathbf 0, \sigma_v^2 \mathbf I_N)$  for $j=1,\cdots,L$.

\subsection{Compression via random projections}
Now assume that the observations at each node are compressed via a low dimensional random projection matrix. The compressed measurement vector at the $j$-th node is given by
\begin{eqnarray}
 \mathbf y_j &=& \mathbf A_j\mathbf x_j \label{obs_comp}
\end{eqnarray}
for $j=1,\cdots, L$ where    $\mathbf A_j$ is an  $M\times N$ ($M< N$) matrix. We assume that the elements of $\mathbf A_j$ are selected so that $\mathbf A_j$ is an orthoprojector; i.e.,  $\mathbf A_j\mathbf A_j^T  = \mathbf I_M$. Let $\mathbf B_j = \mathbf A_j \boldsymbol\Psi$. When $\boldsymbol\Psi$ is an orthonormal basis, we have  $\mathbf B_j\mathbf B_j^T  = \mathbf I_M$ for $j=1, \cdots, L$. The goal is to decide between hypotheses   $\mathcal H_1$ and  $\mathcal H_0$ based on (\ref{obs_comp}).

Let  $\mathcal U$  be the set which contains the indices of locations of nonzero   coefficients in  $\mathbf s_j$ which is defined by
$
\mathcal U := \{i\in \{1,\cdots,N\}~|~ \mathbf s_j (i) \neq 0 \}
$
where $\mathbf s_j (i)$ denotes the $i$-th element of $\mathbf s_j$ for $i=1,\cdots, N$ and $j=1,\cdots,L$. Then we have $k=|\mathcal U|$,  where $|\mathcal U|$ denotes the cardinality of   $\mathcal U$. It is noted that $\mathcal U$ is the same for all the signals $\mathbf s_j$ since we consider the  common sparse support model. Further let $\tilde{ \mathbf v}_j=\mathbf A_j \mathbf v_j$ where $\tilde{ \mathbf v}_j\sim\mathcal (\mathbf 0, \sigma_v^2 \mathbf I_M)$ when $\mathbf A\mathbf A^T = \mathbf I_M$. Then, the detection problem in the compressed domain can be expressed as:
\begin{eqnarray}
\mathcal H_1: ~ \mathbf y_j &=&   \mathbf B_j \mathbf s_j + \tilde{ \mathbf v}_j \nonumber\\
\mathcal H_0: ~ \mathbf y_j &=& \tilde{ \mathbf v}_j \label{det_comp}
\end{eqnarray}
for $j=1,\cdots, L$.
\begin{note}
It is worth noting that we consider the noiseless scenario in  (\ref{obs_comp}). If we consider the noisy scenario, the measurement model will be of  the form $\mathbf y _j =\mathbf A_j \mathbf x_j + \mathbf n_j$ with $\mathbf n_j\sim \mathcal N(\mathbf 0, \sigma_n^2)$, and   we end up with  the same detection problem as in (\ref{det_comp}) except that we will have    $\tilde{ \mathbf v}_j\sim (\mathbf 0, (\sigma_v^2+\sigma_n^2) \mathbf I)$. Thus,  the analysis in the remainder of the paper  will remain the same except that $\sigma_v^2$  will be replaced by $ \sigma_v^2+\sigma_n^2$. For simplicity of presentation, we will employ the noiseless measurement model. Results for noisy measurements can be obtained by changing the value of variance.
\end{note}

\subsection{Sparse signal detection when the common support of the sparse signals  is known}
When  $\mathcal U$ is exactly known, the detection problem in (\ref{det_comp}) reduces to
\begin{eqnarray}
\mathcal H_1: ~ \mathbf y_j &=& \mathbf B_j(\mathcal U)\mathbf s_j(\mathcal U) +\tilde{ \mathbf v}_j\nonumber\\
\mathcal H_0: ~ \mathbf y_j &=&\tilde{ \mathbf v}_j \label{obs_support_known}
\end{eqnarray}
for $j=1,\cdots,L$ where $ \mathbf B_j(\mathcal U)$ denotes the $M\times k$ submatrix of $\mathbf B_j$ in which columns are indexed by the ones in $\mathcal U$, and $\mathbf s_j(\mathcal U)$ is a $k\times 1$ vector containing nonzero elements in $\mathbf s_j$ indexed by $\mathcal U$ for $j=1,\cdots,L$.  When  $ \mathbf B_j(\mathcal U)$  is known, (\ref{obs_support_known}) is the subspace detection problem which has been addressed previously \cite{Scharf_TSP94,Scharf_ASAP03,jin1}. Depending on how the unknown coefficient vector $\mathbf s_j(\mathcal U)$ is modeled, different detectors have been  proposed. In  \cite{Scharf_TSP94}, a generalized likelihood ratio test (GLRT) based detector is proposed when $\mathbf s_j(\mathcal U)$ is assumed to be deterministic. In \cite{Scharf_ASAP03}, the analysis has been  extended to the case  when $\mathbf s_j(\mathcal U)$ is modeled as  random. The problem with multiple observation vectors  is addressed in \cite{jin1} where the authors have proposed adaptive subspace detectors   when the coefficients  $\{\mathbf s_j(\mathcal U)\}_{j=1}^L $ follow first and second order Gaussian models.

In the case of sparse signal detection, it is unlikely that the exact knowledge of $\mathcal U$ is available \emph{a priori}.  In other words,  sparse signal detection needs to be performed when  $\mathcal U$ is unknown. With the advancements of CS,  some algorithms have been developed to detect sparse signals based on  (\ref{det_comp}) exploiting  algorithms developed for sparse signal recovery \cite{haupt_ICASSP07,duarte_ICASSP06,Wimalajeewa_ICASSP13,Gang_globalsip14}.  In particular,  the standard OMP algorithm was modified in \cite{duarte_ICASSP06} to detect the presence of a sparse signal based on a single measurement vector. There, the detection decision is  obtained  after running a few iterations ($\leq k$) of the OMP algorithm. Let $\hat {\mathbf s}_{\mathcal U_{t_0}}$ be the estimated signal after running $t_0\leq k$ number of iterations. Then,  the decision statistic is taken to be the maximum absolute component of $\hat {\mathbf s}_{\mathcal U_{t_0}}$ ($||\hat {\mathbf s}_{\mathcal U_{t_0}}||_{\infty}$ where $||\cdot||_p$ denotes the $p$-norm) in \cite{duarte_ICASSP06}. In \cite{Wimalajeewa_ICASSP13}, a heuristic algorithm is proposed for sparse signal detection in a decentralized manner based on partial support set estimation via OMP  at individual nodes.

Unlike in complete sparse signal recovery,  in  sparse signal detection it is important to  focus on extracting   a decision statistic without completely reconstructing the signal.   To that end,  our goal is to explore the sparse  signal detection problem   with partially known/estimated   support sets.

\section{Sparse Signal Detection with Known Partial Support}\label{section_knownsupport}
Let us  assume that  the detection task is performed with the knowledge of a fraction of the support set of size  $T_0 < k$.  Let $\mathcal U(T_0)\subset \mathcal U$ denote the set containing known indices of the support set.   We assume  that the detection decision is made by comparing  the total power  of the compressed  signals projected on to the subspace spanned by the known subspace to  a threshold.
  More specifically, the decision statistic is given by,
  \begin{eqnarray}
 \Lambda = \sum_{j=1}^L || \mathbf P_{j,T_0}\mathbf y_j||_2^2,  \label{Lamda_1}
 \end{eqnarray}
where \\
$\mathbf P_{j,T_0}= \mathbf B_j({\mathcal U}(T_0)) \left(\mathbf B_j({\mathcal U}(T_0)) ^T \mathbf B_j({\mathcal U}(T_0)) \right)^{-1} \mathbf B_j({\mathcal U}(T_0)) ^T$ is the projection operator to the space spanned by $\mathbf B_j({\mathcal U}(T_0))$. When $T_0 = k$, the decision statistic in (\ref{Lamda_1}) is the same as the GLRT decision  statistic in  \cite{Scharf_ASAP03}.
Under the assumption that  $\mathbf A_j \mathbf A_j^T = \mathbf I$ for $j=1,\cdots,L$, $\Lambda$ in (\ref{Lamda_1}) is distributed  under $\mathcal H_0$ and $\mathcal H_1$ as:
 \begin{eqnarray*}
 \frac{\Lambda}{\sigma_v^2}| \mathcal H_0 &\sim& \mathcal X_{T_0 L }^2\nonumber\\
 \frac{\Lambda}{\sigma_v^2}| \mathcal H_1  &\sim& \mathcal X_{T_0 L}^2(\lambda_{T_0})
   \end{eqnarray*}
respectively,  with $\lambda_{T_0} = \sum_{j=1}^L \frac{||\mathbf P_{j,T_0} \mathbf B_j \mathbf s_j||_2^2}{\sigma_v^2}$.
Let the decision be made by comparing $\Lambda$ in (\ref{Lamda_1}) to a threshold $\tau_0$. Then, the probabilities  of false alarms and detection  can be expressed as \cite{poor1},
 \begin{eqnarray}
 P_f = 1 - F\left(\frac{T_0 L }{2}, \frac{\tau_0}{2\sigma_v^2}\right)\label{P_f_SOMP}
 \end{eqnarray}
 and
  \begin{eqnarray}
 P_d = Q_{\frac{T_0L}{2}} \left(\sqrt{\lambda_{T_0}}, \sqrt{\frac{\tau_0}{\sigma_v^2}}\right), \label{P_d_SOMP}
  \end{eqnarray}
 respectively, where $F(a,x) = \frac{\gamma(a,x)}{\Gamma(a)}$ is the regularized Gamma function,  $\Gamma(\cdot)$ is the Gamma function,  $\gamma(a,x)$ is the lower incomplete Gamma function given by $\gamma(a,x) =\int_0^x t^{a-1}e^{-t} dt $, and  $Q_M(a,b) = \int_b^\infty x\left(\frac{x}{a}\right)^{M-1} e^{-\frac{x^2+a^2}{2}} I_{M-1}(ax) dx$ is the Marcum Q function with $I_{M-1}$ being the modified Bessel function of order $M-1$. Defining the Rayleigh quotient of $\mathbf B_j \mathbf s_j$ with respect to $\mathbf P_{j,T_0}$, $\kappa_{j,T_0}$, as,
\begin{eqnarray}
\kappa_{j,T_0} =\frac{||\mathbf P_{j,T_0} \mathbf B_j \mathbf s_j||_2^2}{|| \mathbf B_j \mathbf s_j||_2^2}, \label{Rayleigh_2}
\end{eqnarray}
 $\lambda_{T_0}$ can be written as,
\begin{eqnarray}
\lambda_{T_0} = \sum_{j=1}^L \frac{|| \mathbf B_j \mathbf s_j||_2^2}{\sigma_v^2} \kappa_{j,T_0}. \label{lamda_t_def}
\end{eqnarray}
When the elements of $\mathbf A$ are random variables  with mean zero and $\mathbf A_j\mathbf A_j^T = \mathbf I$, we may approximate $|| \mathbf B_j \mathbf s_j||_2^2 \approx \frac{M}{N}|| \mathbf s_j||_2^2$. Defining $\gamma_j =  \frac{||  \mathbf s_j||_2^2}{\sigma_v^2} $, (\ref{lamda_t_def}) can be approximated by,
\begin{eqnarray}
\lambda_{T_0} \approx \frac{M}{N}\sum_{j=1}^L \gamma_j \kappa_{j,T_0}.\label{lamda_t_approx}
\end{eqnarray}

Note that,
\begin{eqnarray*}
||\mathbf P_{j,T_0} \mathbf B_j \mathbf s_j||_2^2= || \mathbf B_j \mathbf s_j||_2^2 - ||\mathbf P_{j,T_0}^{\bot} \mathbf B_j \mathbf s_j||_2^2
\end{eqnarray*}
where $\mathbf P_{j,T_0}^{\bot} = \mathbf I - \mathbf P_{j,T_0}$.
Thus, $\kappa_{j,T_0}$ in (\ref{Rayleigh_2}) can be written as,
\begin{eqnarray}
\kappa_{j,T_0} = 1 - \frac{||\mathbf P_{j,T_0}^{\bot} \mathbf B_j \mathbf s_j||_2^2}{||\mathbf B_j \mathbf s_j||_2^2}. \label{Rayleigh_quaotient}
\end{eqnarray}
The quantity $||\mathbf P_{j,T_0}^{\bot} \mathbf B_j \mathbf s_j||_2^2$ reflects the power  of the sampled signal unaccounted for by the subspace spanned by $\mathbf B_j({\mathcal U}(T_0))$. Thus, the impact of not having  the knowledge of the  complete support set of $\mathbf s_j$  on the detection performance is reflected by $||\mathbf P_{j,T_0}^{\bot} \mathbf B_j \mathbf s_j||_2^2$.
When  $T_0 = k$,   we have $||\mathbf P_{j,T_0}^{\bot} \mathbf B_j \mathbf s_j||_2^2 = 0$ and  $\kappa_{j,T_0} = 1$ and, therefore,  $\lambda_{k}  \approx \frac{M}{N}\sum_{j=1}^L \gamma_j $. Let $\tau_d$ be the desired probability of detection.  If  $Q_{\frac{kL}{2}} \left(\sqrt{\lambda_{k}}, \sqrt{\tau_0/\sigma_v^2}\right) < \tau_d$, the desired $P_d$ with the decision  statistic (\ref{Lamda_1}) cannot be achieved even if all the  indices of the common support are known correctly. Thus,  estimation of only a fraction is sufficient only if the desired  detection performance is such that  $\tau_d < Q_{\frac{kL}{2}} \left(\sqrt{\lambda_{k}}, \sqrt{\tau_0/\sigma_v^2}\right)$. In that case, it is of interest to determine the minimum fraction of the support to be known in order to achieve the desired detection performance which will be discussed next.

The  goal is to find the minimum value of $T_0$ in order to achieve a desired $P_d$ while maintaining  the probability of false alarm  $P_f$  under a desired value, say, $\alpha$.  Letting $t=T_0$, the desired optimization problem can be cast as,
 \begin{eqnarray}
&\min t ~ \mathrm{such} ~\mathrm{that}& \nonumber\\
& P_d(t) \geq \tau_d, P_f \leq \alpha ~ \mathrm{and} ~ 1 \leq t <  k, t\sim \mathrm{integers} &  \label{opt_OMP}
\end{eqnarray}
with
\begin{eqnarray}
P_d(t) = Q_{\frac{tL}{2}} \left(\sqrt{\lambda_{t}}, \sqrt{\tau_0/\sigma_v^2}\right) \label{Pd_marcm_Q}
 \end{eqnarray}
 where $\lambda_t $ is as given in (\ref{lamda_t_approx}). The term  $t$ is related to $\lambda_t$ via  $\kappa_{j,t} = 1 - \frac{||\mathbf P_{j,t}^{\bot} \mathbf B_j \mathbf s_j||_2^2}{||\mathbf B_j \mathbf s_j||_2^2}$. Let $\mathbf P_{j,t}^{\bot} = \mathbf U^T \boldsymbol\Pi \mathbf U$ where $\mathbf U$ is unitary and $\boldsymbol\Pi$ is a diagonal matrix with $t$ zero elements and $(M-t)$ ones. Then, we can write
\begin{eqnarray}
||\mathbf P_{j,t}^{\bot} \mathbf B_j \mathbf s_j||_2^2 &=& ||\mathbf P_{j,t}^{\bot} \mathbf B_{j,k\setminus t} \mathbf s_{j,k\setminus t} ||_2^2 \nonumber\\
&=&||\boldsymbol\Pi \mathbf U\mathbf B_{j,k\setminus t} \mathbf s_{j,k\setminus t} ||_2^2 \label{P_j_bot}
\end{eqnarray}
where $\mathbf B_{j,k\setminus t} = \mathbf B_j(\mathcal U\setminus \mathcal U(t))$ with $|\mathcal U|=k$ (similar definition holds for $\mathbf s_{j,k\setminus t}$). Whenever  $\mathbf B_j^T \mathbf B_j \approx  \frac{M}{N}\mathbf I$, we have   $\breve {\mathbf B}_j^T \breve {\mathbf B}_j \approx \frac{M}{N} \mathbf I$  since $\mathbf U$ is unitary  where $\breve {\mathbf B}_j = \mathbf U\mathbf B_j$.
%Note that, the elements of $\mathbf B_j$ are such that $\mathbf B_j \mathbf B_j^T = \mathbf I$ and $\mathbf B_j^T \mathbf B_j = \frac{M}{N}\mathbf I$.
Since $\boldsymbol\Pi$ has only $M-t$ ones, we may approximate (\ref{P_j_bot}) by,
\begin{eqnarray*}
||\boldsymbol\Pi \mathbf U\mathbf B_{j,k\setminus t} \mathbf s_{j,k\setminus t} ||_2^2 &=& ||\boldsymbol\Pi \breve {\mathbf B} _{j,k\setminus t} \mathbf s_{j,k\setminus t} ||_2^2 \nonumber\\
&\approx& \frac{M-t}{N} || \mathbf s_{j,k\setminus t} ||_2^2.
\end{eqnarray*}
Then, $\kappa_{j,t}$ can be approximated by
\begin{eqnarray}
\kappa_{j,t}\approx 1 - \frac{(M-t)}{M} \frac{|| \mathbf s_{j,k\setminus t} ||_2^2}{||\mathbf s_{j}||^2}. \label{kappa_2}
\end{eqnarray}
It is noted that  $|| \mathbf s_{j,k\setminus t} ||_2^2$ depends on  $t$ and the nonzero coefficients of the signal indexed by $\mathcal U\setminus\mathcal U(t)$.  Since the nonzero coefficients of $\mathbf s_j$ are unknown,  the  knowledge of $|| \mathbf s_{j,k\setminus t} ||_2^2$ for given $t$ is not available. Thus, in the following we solve the problem under certain assumptions regarding the  nonzero coefficients. We discuss the impact of relaxations of these assumptions in the numerical results section.

\subsubsection*{Nonzero coefficients do  not significantly deviate  from each other}
In the case where the nonzero coefficients do not deviate much from each other, we may approximate $\frac{|| \mathbf s_{j,k\setminus t} ||_2^2}{|| \mathbf s_{j} ||_2^2}\approx \frac{(k-t)}{k}$. In that case, (\ref{kappa_2}) can be approximated by
\begin{eqnarray}
\kappa_{j,t}\approx \frac{t}{k}\left(1+\frac{k-t}{M}\right). \label{kappa_3}
\end{eqnarray}
Then  we have
\begin{eqnarray}
{\lambda}_{t} \approx  \frac{Mt}{Nk}\left(1+\frac{(k-t)}{M}\right) \sum_{j=1}^L \gamma_j.\label{lamda_t}
\end{eqnarray}

With ${\lambda}_{t}$ as in (\ref{lamda_t}), we aim to solve (\ref{opt_OMP}). Note  that (\ref{opt_OMP}) is an integer programming problem with linear and nonlinear constraints. To further simplify the nonlinear constraint $P_d(t)\geq \tau_d$,   we use certain approximations for  $P_d$ and $P_f$  in (\ref{P_f_SOMP}) and (\ref{P_d_SOMP}), respectively,  exploiting  approximations of the tail probabilities of chi-squared  random variables.

The cumulative distribution function (cdf), $G(x,k)$,  of a chi-squared  random variable, $x\sim\mathcal X^2_k$, can be approximated by  $G(x,k) \approx 1- Q \left(\frac{\left(\frac{x}{k}\right)^{1/3} - \left(1-\frac{2}{9k}\right)}{\sqrt{\frac{2}{9k}}}\right)$ \cite{Sankaran_Bio63}, where $Q(x) = \frac{1}{\sqrt{2\pi}}\int_x^{\infty} e^{-\frac{t^2}{2}} dt$ is the Gaussian $Q$-function. Using this approximation for $P_f$, we have,
\begin{eqnarray*}
P_f \approx Q\left(\frac{\left(\frac{\tau_0/\sigma_v^2}{tL}\right)^{1/3} - \left( 1- \frac{2}{9tL}\right)}{\sqrt{\frac{2}{9tL}}}\right)
\end{eqnarray*}
and the threshold should be selected as
 \begin{eqnarray}
 \tau_0(t,\alpha) \approx \sigma_v^2 t L  \left(\left(1-\frac{2}{9tL }\right) + \sqrt{\frac{2}{9tL }} Q^{-1}(\alpha)\right)^3\label{tau_0_OMP}
 \end{eqnarray}
 to maintain $P_f \leq \alpha$.
 Similarly, the cdf of a non-central chi-squared random variable, $x\sim \mathcal X_k^2(\lambda)$, can be approximated by  $G(x;k,\lambda) \approx 1- Q\left(\frac{\left(\frac{x}{k+\lambda}\right)^h - (1+hp(h-1-0.5(2-h)mp))}{h\sqrt{2p}(1+0.5mp)}\right)$   where $h = 1 - \frac{2}{3} \frac{(k+\lambda)(k+3\lambda)}{(k+2\lambda)^2}$, $p=\frac{k+2\lambda}{(k+\lambda)^2}$,  and $m=(h-1)(1-3h)$  \cite{Sankaran_Bio63}. With this approximation,  $P_d$ can be written  as,
\begin{eqnarray}
P_d \approx Q\left( f(t, \alpha, {\boldsymbol \gamma})\right) \label{P_d_OMP}
\end{eqnarray}
with
\begin{eqnarray}
&~& f(t, \alpha,{\boldsymbol \gamma}) \nonumber\\
&=& \left(\frac{\left(\frac{\tau_0(t,\alpha)}{\sigma_v^2(t L+{\lambda}_t)}\right)^h - (1+hp(h-1-0.5(2-h)mp))}{h \sqrt{2p}(1+0.5 mp)}\right) \label{f_t}
\end{eqnarray}
where  $h=1-\frac{2}{3} \frac{(t L + {\lambda}_{t})(t L + 3{\lambda}_{t})}{(t L + 2{\lambda}_{t})^2}$, $p = \frac{t L + 2{\lambda}_{t}}{(t L + {\lambda}_{t})^2}$ and $m=(h-1)(1-3h)$, ${\boldsymbol \gamma} = [\gamma_1,\cdots,\gamma_L]$ and $\tau_0(t, \alpha)$ is approximated as in (\ref{tau_0_OMP}). It is noted that $\lambda_t$ in the right hand side of  (\ref{f_t}) is given by (\ref{lamda_t}) which is a function of $\boldsymbol\gamma$.

With the approximations  (\ref{lamda_t}),  (\ref{tau_0_OMP}) and   (\ref{P_d_OMP}),  the optimization problem in (\ref{opt_OMP}) can be rewritten as,
\begin{eqnarray}
&\min t ~\mathrm{such}~ \mathrm{that}~& \nonumber\\
& \tau_d - Q\left(f(t,\alpha, {\boldsymbol \gamma})\right)  \leq 0,&\nonumber\\
& ~t-k <  0  ~\mathrm{and } ~1 - t\leq 0, t \sim\mathrm{integers}. & \label{opt_3}
\end{eqnarray}
The theory
developed in the area of integer nonlinear programming  is much less mature
than integer linear programming \cite{Junger_2010}.  We aim to solve (\ref{opt_3}) based on rounding by relaxing the integer restriction. This is a heuristic method which is shown to be faster than most of the other algorithms developed for integer nonlinear programming \cite{Leyffer_PhD1993}.   After relaxing the integer restriction, the solution for $t$ which satisfies the  Karush-Kuhn-Tucker  (KKT) conditions is summarized in the following proposition.
\begin{proposition}\label{proposition}
Let $k>1$. The solution for $t$ in (\ref{opt_3}) which satisfies the KKT conditions is given by  (\ref{hat_u})
\begin{figure*}[!t]
\begin{eqnarray}
\hat t_{cont}=\left\{
\begin{array}{ccc}
1~&\mathrm{if} ~  \tau_d \leq  Q(f(1))\\
f^{-1}(Q^{-1}(\tau_d)) &~\mathrm{if}~ f'(t)|_{t=f^{-1}(Q^{-1}(\tau_d))} \leq 0 ~\& ~Q(f(1)) <  \tau_d < Q(f(k))  \\
%k  &~\mathrm{if}~ \tau_d = Q(f(k)) \\
\mathrm{infeasible} & \mathrm{otherwise}
\end{array}\right.\label{hat_u}
\end{eqnarray}
\end{figure*}
where $\hat t_{cont}$ is the continuous valued solution for $t$.
\end{proposition}
\begin{proof}
For $\mu_1, \mu_2, \mu_3\geq 0$, the Lagrangian for  (\ref{opt_3}) is given by,
\begin{eqnarray*}
L(t, \mu) = t+ \mu_1(\tau_d - Q(f(t))) + \mu_2(t-k) + \mu_3(1-t)
\end{eqnarray*}
where we use $f(t)$ for  $f(t,\alpha, {\boldsymbol \gamma})$ for simplicity.
The KKT conditions are given by,
\begin{eqnarray}
1 - \mu_1 \frac{\partial Q(f(t))}{\partial t} + \mu_2-\mu_3 = 0\label{kkt_1}\\
\mu_1(\tau_d - Q(f(t))) = 0,
\mu_2 (t-k) = 0 ,
\mu_3(1-t) = 0 \label{kkt_2}\\
\tau_d - Q(f(t)) \leq 0, t-k <  0, 1-t \leq 0. \label{kkt_3}
\end{eqnarray}
Using the Leibnitz integral rule, $\frac{\partial Q(f(t))}{\partial t}$ can be computed as,
\begin{eqnarray}
\frac{\partial Q(f(t))}{\partial t} = - \frac{1}{\sqrt{2\pi}} e^{-\frac{f(t)^2}{2}} f'(t)\label{partial_fu}
\end{eqnarray}
where $f'(t)$ is the derivative of $f(t)$ with respect to $t$. With $\frac{\partial Q(f(t))}{\partial t}$ as in (\ref{partial_fu}), it can be verified that (\ref{kkt_1})-(\ref{kkt_3})  have a feasible solution only under the  following cases: (i). $\mu_1 = 0$, $\mu_2=0$ and $\mu_3 \neq 0$, (ii). $\mu_1 \neq 0$, $\mu_2=0$ and $\mu_3 = 0$, (iii). $\mu_1 \neq  0$, $\mu_2 = 0$ and $\mu_3 \neq 0$, and (iv). $\mu_1 \neq 0$, $\mu_2\neq 0$ and $\mu_3 =0$. Then,  $\hat t_{cont}$ which  satisfies (\ref{kkt_1})-(\ref{kkt_3}) can be obtained as in (\ref{hat_u}).
\end{proof}

The integer valued solution, $\hat t$,  is obtained  by rounding $\hat t_{cont}$ to the nearest integer.  From Proposition \ref{proposition}, it is observed that, under certain conditions,  i.e.,  when the parameters ${\boldsymbol\gamma}$, $\alpha$ and $\tau_d$ are such that $ \tau_d \leq  Q(f(1, \alpha,\boldsymbol{\gamma}))$, it is sufficient to know   only one support location  of the sparse signal correctly to reach the desired detection performance.  It is further interesting to investigate the infeasible case. In Fig. \ref{fig:fdasht}, we plot $f'(t)$ vs $t$. While  $f'(t)$ depends on several parameters,  we show the behavior of $f'(t)$  for a given set of  values for  $N$, $c_r=M/N$,  and $L$ as $k$, and $\sigma_v^2$ (and thus $\gamma_j$) vary. In order to compute $\gamma_j$, we generate the nonzero coefficients of $\mathbf s_j$'s from a uniform distribution in $[3,4]$. For the  sparse support,  $k$ indices are selected randomly (and uniformly)  from $[1,N]$. It is observed from Fig. \ref{fig:fdasht} that the infeasible case can be observed when the signal is less sparse (i.e. $k$ is large). As seen in Fig.  \ref{fig:fdasht} (b) after some value of $t$, $f'(t)$ becomes  positive. In other words, when $t$ exceeds this threshold,  $P_d(t)$ starts decreasing. When the  signal is sufficiently sparse (i.e. $k\ll N$), the  infeasible case is less likely to be observed.

\begin{figure*}
    \centering
    \begin{subfigure}[b]{0.45\textwidth}
        \includegraphics[width=\textwidth]{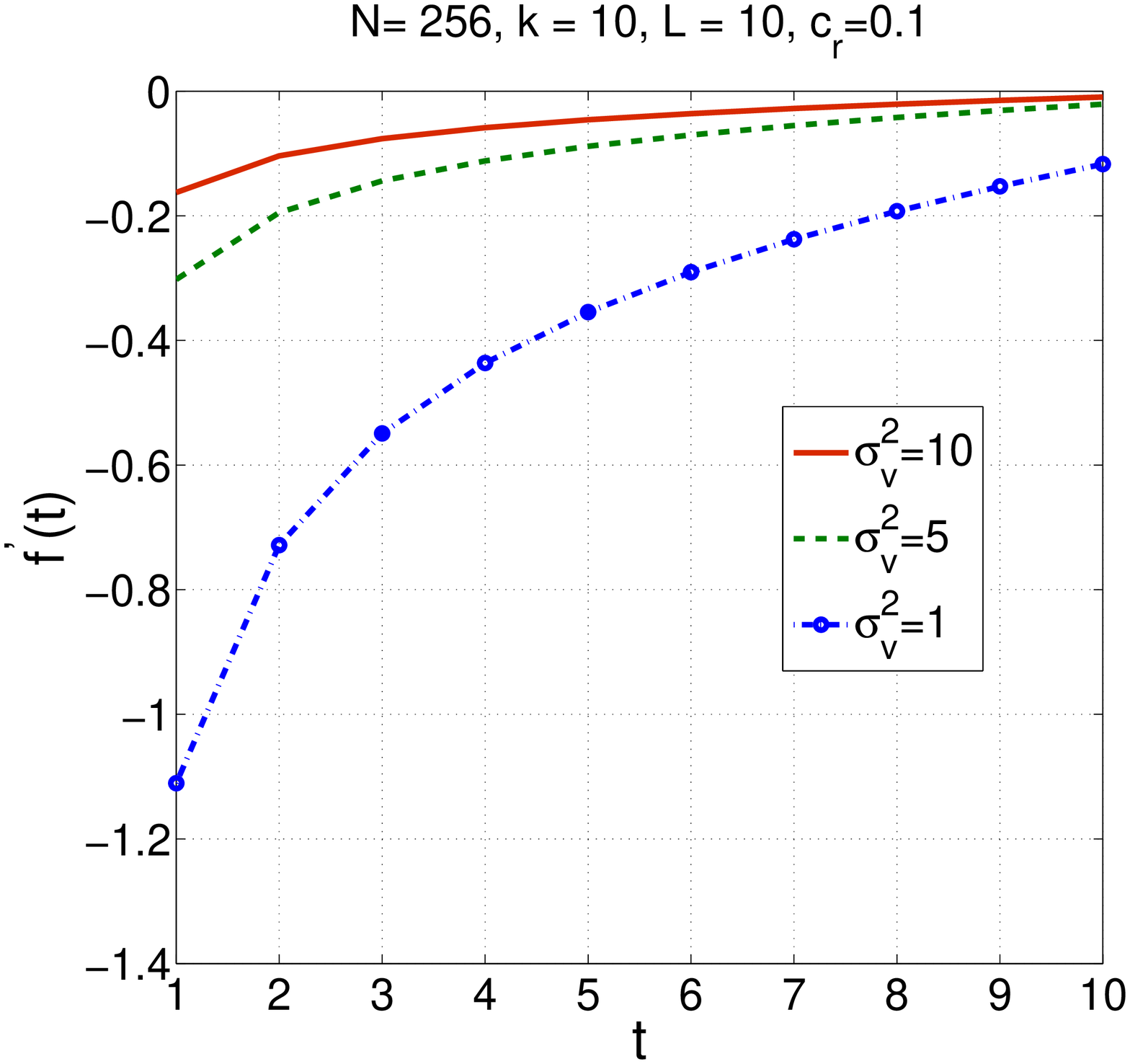}
        \caption{ $K=10$}
        \label{fig:gull}
    \end{subfigure}
 \begin{subfigure}[b]{0.45\textwidth}
        \includegraphics[width=\textwidth]{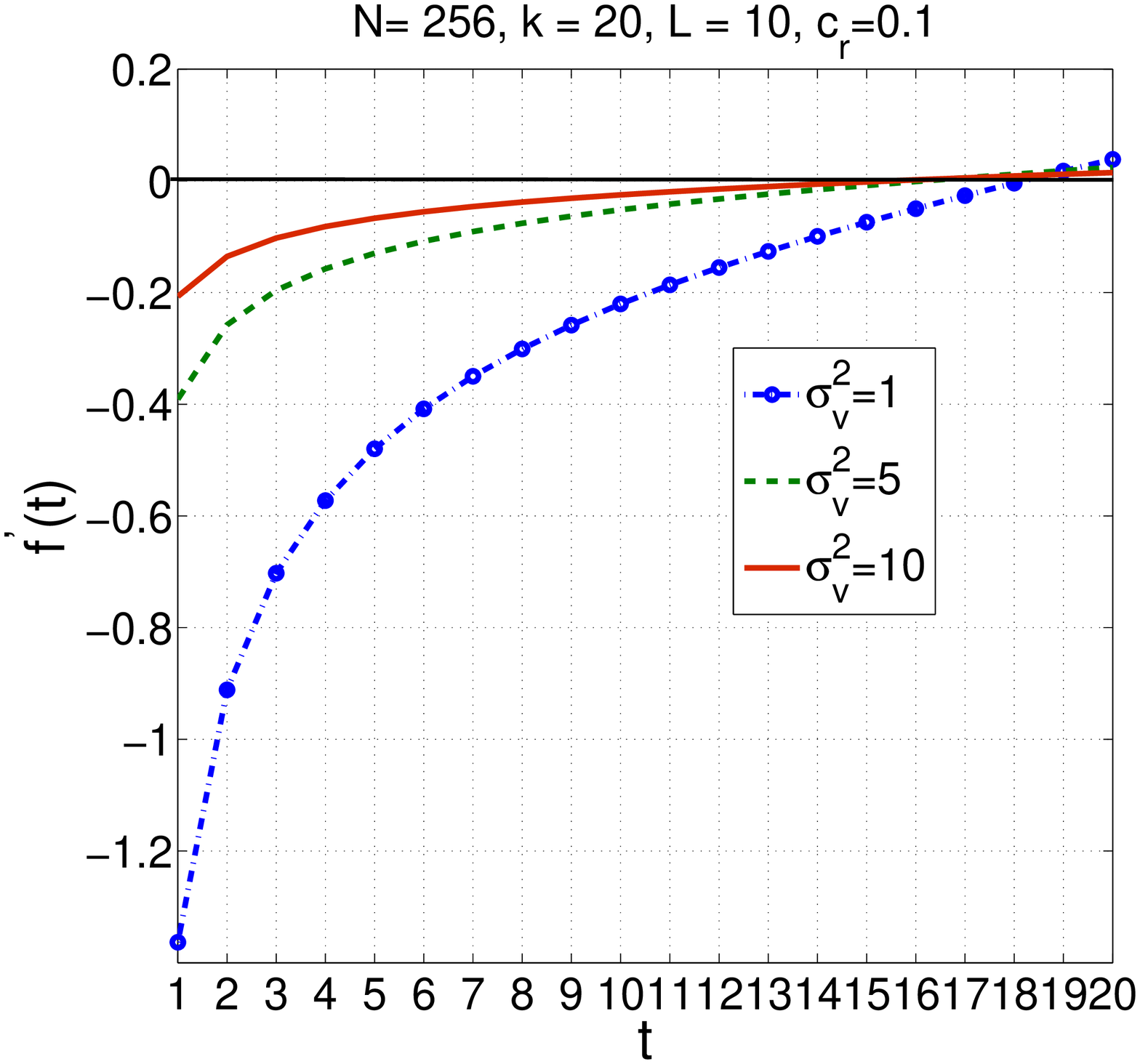}
        \caption{ $K=20$}
        \label{fig:tiger}
    \end{subfigure}
    \caption{Behavior of $f'(t)$ and $P_d(t)\approx Q(f(t))$; $N=256$, $L=10$, $c_r=0.1$}\label{fig:fdasht}
\end{figure*}

The above analysis provides insights on sparse signal detection when the  minimum size of the partial support set is computed under certain assumptions on the sparse signals. In practice, the desired  fraction of the support of the sparse signal   needs to be estimated based on  the available data when it is not  known  \emph{a priori} whether  the signal is present or not.   In the following section, we consider several extensions of the OMP algorithm to detect the presence of a sparse signal by estimating the partial support of given size  in centralized as well as distributed settings.

\section{OMP Based  Sparse Signal Detection via Partial Support Set  Estimation}\label{sec_OMP_detection}
OMP is a greedy algorithm developed for sparse signal recovery with a single measurement vector (SMV) \cite{tropp_OMP2007}. In Algorithm \ref{algo_OMP}, we state the standard OMP for sparse support set estimation (the coefficients can be estimated using least squares  estimation). It is noted that,  the subscripts of the vectors and matrices  corresponding to node index are dropped for brevity.
 \begin{algorithm}
%\setstretch{0.5}
\textbf{Inputs}: $\mathbf y$, $\mathbf B$, $k$\\
\textbf{Output}: An estimate for the support set, $\hat{\mathcal U}$
\begin{enumerate}
\item Initialize $t=1$, $\hat{\mathcal U}(0) = \emptyset $, residual vector $\mathbf r_{0} = \mathbf y$
\item Find the index $\beta(t)$ such that
$
\beta(t) = \underset{\omega = 1,\cdots,N}{\arg~ \max} ~ {|\langle\mathbf r_{t-1}, \mathbf B({\omega})\rangle|}
$
 %\item Update the index set $\beta_l^*(t)$ via local communication: $\beta_l^*(t) = f_l(\beta_l(t),\{\beta_{i}(t)\}, i\in \mathcal M_l)$, as discussed in subsection \ref{step_3}
 \item  Set $\hat{\mathcal U}(t) = \hat{\mathcal U}(t-1) \cup \{\beta (t)\}$
\item  Compute the projection operator $\mathbf P(t) = \mathbf B(\hat{\mathcal U}(t)) \left(\mathbf B(\hat{\mathcal U}(t)) ^T \mathbf B(\hat{\mathcal U}(t)) \right)^{-1} \mathbf B(\hat{\mathcal U}(t)) ^T$. Update the residual vector:  $\mathbf r_{t} = (\mathbf I - \mathbf P(t))\mathbf y$ (note: $\mathbf B(\hat{\mathcal U}(t))$ denotes the submatrix of $\mathbf B$ in which columns are taken from $\mathbf B$ corresponding to the indices in $\hat{\mathcal U}(t)$)
     \item  Increment $t=t+1$ and go to step 2 if $t\leq k$, otherwise, stop and set $\hat{\mathcal U} = \hat{\mathcal U}(t-1)$

 \end{enumerate}
 \caption{Standard OMP  for sparse support set estimation with SMV}\label{algo_OMP}
 \end{algorithm}

The OMP algorithm was extended to the multiple measurement vector case in \cite{Tropp_P12006,Baron1} which is termed simultaneous OMP (S-OMP).
First,  we consider that all the compressed observations are available at the fusion  center. In  Algorithm \ref{algo_SOMP_det},   we extend  the S-OMP algorithm for sparse signal detection by first estimating a partial support set of size $T_0$ based on compressed observations  in (\ref{obs_comp}). This is a simple modification of  the S-OMP algorithm.
\begin{algorithm}
%\setstretch{0.5}
\textbf{Inputs}:  $\{\mathbf y_j  \}_{j=1}^{L}, \{\mathbf B\}_{j=1}^{L}$, $T_0$\\
\textbf{Outputs}: Partial support set estimate $\hat{\mathcal U}(T_0)$, Decision statistic $\Lambda_{cent}$, Detection decision
\begin{enumerate}
\item Initialize $t=1$, $\hat{\mathcal U}(0) = \emptyset $, residual vector $\mathbf r_{j,0} = \mathbf y_j$ for $j=1, \cdots, L$
   \item    \emph{For} $t=1$ to $t=T_0$
\item  Find the index $\beta(t)$ such that
 \begin{eqnarray*}
\beta(t) = \underset{\omega = 1,\cdots,N}{\arg~ \max} ~ \sum_{j=1}^{L}{{|\langle\mathbf r_{j,t-1}, \mathbf B_{j}(\omega)\rangle| }}
\end{eqnarray*}
 \item   Set $\hat{\mathcal U}(t) = \hat{\mathcal U}(t-1) \cup \{\beta(t)\}$
\item Compute the orthogonal projection operator:
\begin{scriptsize}
$\mathbf P_{j,t} =  \mathbf B_j(\hat{\mathcal U}(t)) \left(\mathbf B_j(\hat{\mathcal U}(t)) ^T \mathbf B_j(\hat{\mathcal U}(t)) \right)^{-1} \mathbf B_j(\hat{\mathcal U}(t)) ^T $\end{scriptsize} \hspace{1cm} for $j=1, \cdots, L$\\
Update the residual:  $\mathbf r_{j,t} = (\mathbf I - \mathbf P_{j,t})\mathbf y_j$ for $j=1,\cdots, L$
\item  End \emph{For}
     \item Set    $\hat{\mathcal U}(T_0) = \hat{\mathcal U}(t)$
\item Detection decision:  \\
If $\Lambda_{cent} = \sum_{j=1}^L || \mathbf P_{j,T_0} \mathbf y_j||_2^2 \geq \tau_0$, $\mathcal H_1$ is true, otherwise $\mathcal H_0$ is true where $\tau_0$ is the threshold.
 \end{enumerate}
 \caption{S-OMP based sparse signal detection: centralized approach}\label{algo_SOMP_det}
 \end{algorithm}

In a  centralized setting, each node has to transmit its length-$M$ compressed measurement vector  to the fusion center so that the fusion center processes $\{\mathbf y_l, \cdots, \mathbf y_L\}$ to make the detection decision. While this reduces the communication burden compared to forwarding uncompressed  data vectors of length-$N$, in the following  we consider further reduction  of information  to be transmitted by each node.    We propose two  distributed algorithms and they differ from each other in terms of  the communication overhead.
\begin{algorithm}
\textbf{Inputs}:  (At the $j$-th node) $\mathbf y_j  , \mathbf B_j$, $T_0$ for $j=1,\cdots,L$\\
\textbf{Outputs}: (At the $j$-th node) Partial support set estimate $\hat{\mathcal U}_j$, (At the fusion center) Decision statistic $\Lambda_{dist1}$, Detection decision\\
\textbf {Initialization}:\\
At the $j$-th node: $\hat{\mathcal U}_j(0) = \emptyset $, residual vector $\mathbf r_{j,0} = \mathbf y_j$ for $j=1,\cdots,L$.
\begin{enumerate}
 \item[-] \textbf {At the $j$-th node for $j=1,\cdots, L$}
 \item \emph{For} $t=1$ to $t=T_0$
\item Find the index $\beta_j(t)$ such that
$
\beta_j(t) = \underset{\omega = 1,\cdots,N}{\arg~ \max} ~ {|\langle \mathbf r_{j,t-1}, \mathbf B_j(\omega)\rangle|}
$
\item  Set $\hat{\mathcal U}_j(t) = \hat{\mathcal U}_j(t-1) \cup \{\beta_j (t)\}$
\item  Compute the projection operator $\mathbf P_{j,t} = \mathbf B_j(\hat{\mathcal U}_j(t)) \left(\mathbf B_j(\hat{\mathcal U}_j(t)) ^T \mathbf B_j(\hat{\mathcal U}_j(t)) \right)^{-1} \mathbf B_j(\hat{\mathcal U}_j(t)) ^T$. Update the residual vector:  $\mathbf r_{j,t} = (\mathbf I - \mathbf P_{j,t})\mathbf y_j$
    \item End \emph{For}
       \item  Set $\hat{\mathcal U}_j =\hat{\mathcal U}_j(T_0)$ and $ \mathbf P_{j,\hat{\mathcal U}_j}= \mathbf P_{j,T_0}$
\item Compute $\Lambda_j = || \mathbf P_{j,\hat{\mathcal U}_j} \mathbf y_j||_2^2 $ and transmit $\Lambda_j$ to the fusion center
\item[-] \textbf{At the fusion center}
\item Receive $\Lambda_j$ for $j=1, \cdots, L$
    \item Compute the decision statistic $\Lambda_{dist1} = \sum_{j=1}^{L} \Lambda_j$
\item Detection decision: if $\Lambda_{dist1}  \geq \tau_0$ decide $\mathcal H_1$, otherwise decide $\mathcal H_0$
 \end{enumerate}
 \caption{OMP  based sparse signal detection: distributed approach 1}\label{algo_dist_0}
 \end{algorithm}

A simple  version resulting  in  low  communication overhead is to obtain  the length $T_0$ support set independently at each node and compute a local decision statistic which  is then transmitted to the fusion center.  The fusion center then combines the local contributions to obtain the final decision statistic. This algorithm is presented in Algorithm \ref{algo_dist_0}. Based on the partial support set $\hat{\mathcal U}_j$ of size $T_0$ obtained  in step 6, $\Lambda_j = || \mathbf P_{j,\hat{\mathcal U}_j} \mathbf y_j||_2^2 $ is computed and transmitted to the fusion center by the $j$-th node for $j=1,\cdots,L$. Then, the fusion center  computes the  decision statistic $\Lambda_{dist1}$ as given in step 9 in Algorithm \ref{algo_dist_0}. Intuitively, it is expected that Algorithm \ref{algo_SOMP_det} performs better than Algorithm \ref{algo_dist_0}. However, in the  following,  we show that this is not true always (for certain values of $M$ and $T_0$).

\subsection{Comparing Algorithms \ref{algo_SOMP_det} and  \ref{algo_dist_0}}\label{compare_2_algo}
It is  noted that the $j$-th  element in the sum in $\Lambda_{cent}$ in Algorithm \ref{algo_SOMP_det}  accounts for the power of the compressed   observations   projected on to the subspace spanned by the columns of $\mathbf B_j$ for $j=1,\cdots, L$ indexed by the same set $\hat{\mathcal U}(T_0)$. On the other hand, $j$-th element in  the sum in $\Lambda_{dist1}$ in Algorithm \ref{algo_dist_0} represents  the power of the compressed observations  projected on to the subspace spanned by the columns of $\mathbf B_j$  indexed by $\hat{\mathcal U}_j$ which in general can be  different for $j=1,\cdots, L$. Thus, when $M$ is not very large, there can be at least one correct index in $\{\mathcal U_j\}_{j=1}^L$ although $\hat{\mathcal U}(T_0)$ does not contain any correct index, especially when $T_0$ is also small. In that case,  all the elements in the sum in  $\Lambda_{cent}$ correspond to the power of the compressed   observations projected on to a noise subspace while there can be  at least one element in the  sum in $\Lambda_{dist1}$ that accounts for the power projected into the signal subspace leading to better detection performance by Algorithm \ref{algo_dist_0}. We further analyze this scenario when $T_0=1$.

Let $P_1$ be the probability that $\beta(1)$ (to avoid confusion while referring to the two algorithms in the following discussion, we denote this as $\beta^c(1)$) estimated at step 3 in Algorithm \ref{algo_SOMP_det} is a correct index under $\mathcal H_1$. Similarly, let  $P_2$ be the probability that at least one $\beta_j(1)$ (we denote this as $\beta^d_j(1)$) for $j=1,\cdots, L$ estimated at step 2 in Algorithm \ref{algo_dist_0} is correct under $\mathcal H_1$.  Then we have,
\begin{eqnarray*}
P_1 &=& Pr(\beta^c(1)\in \mathcal U| \mathcal H_1) \nonumber\\
&=&Pr \left(\left\{\underset{\omega = 1,\cdots,N}{\arg~ \max} ~ \sum_{j=1}^{L}{{|\langle\mathbf y_{j}, \mathbf B_{j}(\omega)\rangle| }}\right\} \in \mathcal U |\mathcal H_1\right).
\end{eqnarray*}
On the other hand,
\begin{eqnarray}
P_2 &=& Pr(\beta^d_1(1)\in \mathcal U ~\mathrm{or}, \cdots, \mathrm{or}  ~ \beta^d_L(1)\in \mathcal U|\mathcal H_1).\label{eq_P2}
\end{eqnarray}
\begin{figure*}
    \centering
    \begin{subfigure}[b]{0.35\textwidth}
        \includegraphics[width=\textwidth]{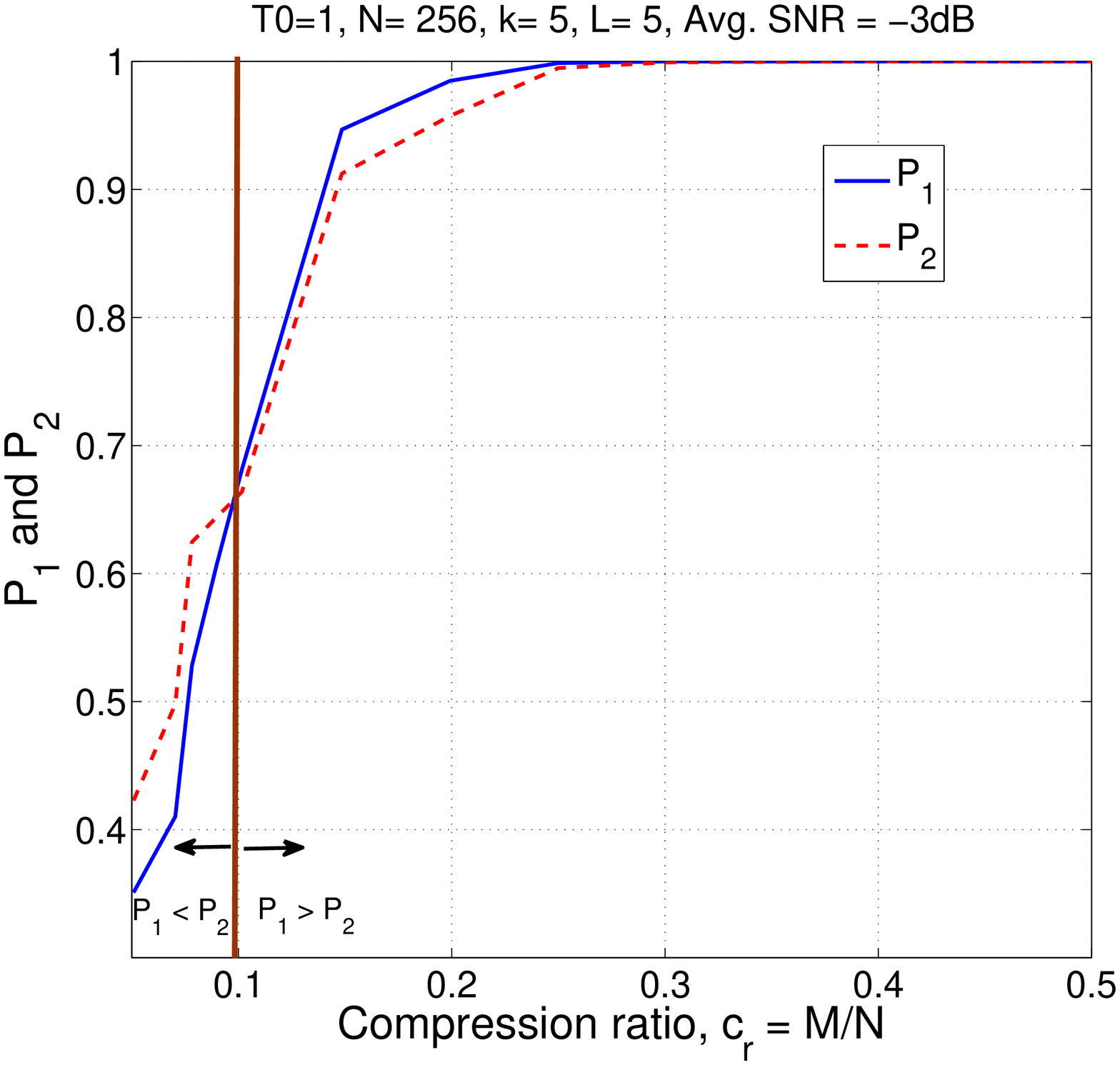}
        \caption{$k=5$, $\mathrm{SNR}=-3dB$}
           \end{subfigure}
    ~ %add desired spacing between images, e. g. ~, \quad, \qquad, \hfill etc.
    \begin{subfigure}[b]{0.35\textwidth}
        \includegraphics[width=\textwidth]{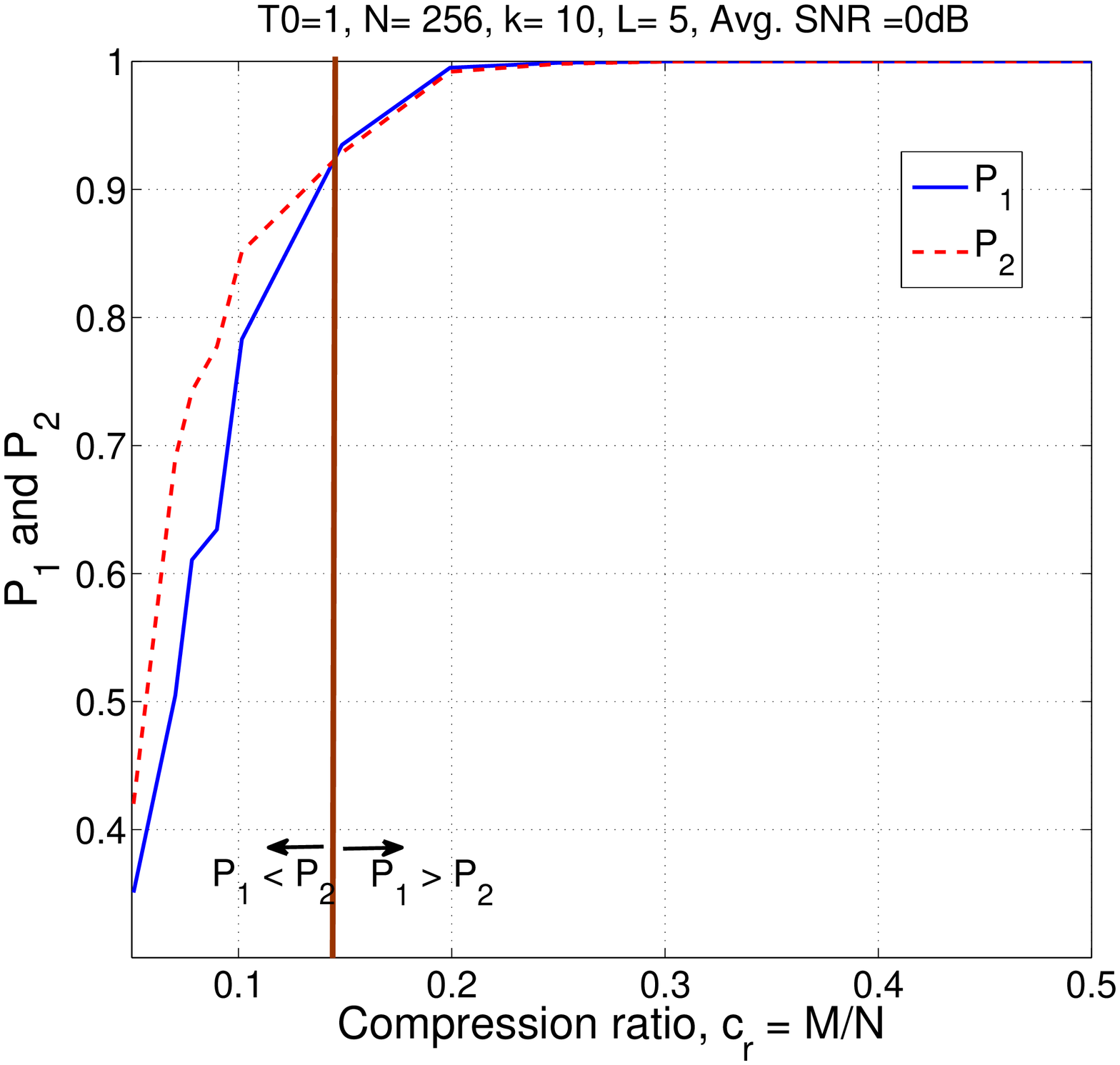}
        \caption{$k=10$, $\mathrm{SNR}=0dB$}
        \end{subfigure}
       \caption{Comparison of $P_1$ in (\ref{P_1}) and $P_2$ in (\ref{P_2}) as $c_r = M/N$ varies}\label{fig:P1_P2}
\end{figure*}
Authors in \cite{tropp_OMP2007} and \cite{Tropp_P12006}, respectively,  consider the conditions under which  OMP (with a single measurement vector) and S-OMP (with multiple measurements vectors)  are  capable of recovering the complete support set   while we focus  here only on  $T_0=1$.
Let us partition $\mathbf B_j$ as  $\mathbf B_j = [\bar{\mathbf B}_{j}~ \tilde{\mathbf B}_j]$ where $\bar{\mathbf B}_{j}$ and $\tilde{\mathbf B}_j$  are submatrices of $\mathbf B_j$ with $\bar{\mathbf B}_{j}$ containing   the columns of $\mathbf B_j$ indexed by $\mathcal U$, while $\tilde{\mathbf B}_j$ contains the rest of the columns of  $\mathbf B_j$ for $j=1,\cdots,L$. Let
\begin{eqnarray}
\rho_j(\mathbf y_j) = \frac{||\tilde{\mathbf B}_j^T \mathbf y_j||_{\infty}}{||\bar{\mathbf B}_j^T \mathbf y_j||_{\infty}}
\end{eqnarray}
where $||\cdot||_{\infty}$ is the infinity norm of a vector (or a matrix).
Then  $P_2$ in (\ref{eq_P2}) can be expressed as,
\begin{eqnarray}
P_2 &=& Pr( \rho_1(\mathbf y_1) < 1 ~\mathrm{or}, \cdots, \mathrm{or}  ~ \rho_L(\mathbf y_L) < 1 |\mathcal H_1)\nonumber\\
&=& 1 - \prod_{j=1}^L Pr(\rho_j(\mathbf y_j) > 1 | \mathcal H_1)\label{P_2}
\end{eqnarray}
where the second equality is due to the fact that the multiple nodes estimate the indices independently.
In order to compute $P_1$, let
\begin{eqnarray}
\rho_c(\mathbf Y) = \frac{||[\tilde{\mathbf B}_1^T \mathbf y_1, \cdots, \tilde{\mathbf B}_L^T \mathbf y_L] ||_{\infty}}{||[\bar{\mathbf B}_1^T \mathbf y_1, \cdots, \bar{\mathbf B}_L^T \mathbf y_L]||_{\infty}}
\end{eqnarray}
where $[\tilde{\mathbf B}_1^T \mathbf y_1, \cdots, \tilde{\mathbf B}_L^T \mathbf y_L]$ denotes the matrix in which the $j$-th  column is $\tilde{\mathbf B}_j^T \mathbf y_j$ for $j=1,\cdots,L$.
Then we have,
\begin{eqnarray}
P_1 = Pr(\rho_c(\mathbf Y) < 1).\label{P_1}
\end{eqnarray}

It is difficult to find exact analytical expressions for $P_1$ and $P_2$ in (\ref{P_1}) and (\ref{P_2}), respectively. Thus, in Fig. \ref{fig:P1_P2} we plot $P_1$ and $P_2$ obtained numerically  as the compression ratio $c_r = \frac{M}{N}$ varies for different values of $k$ and $L$. It can be observed that there is a threshold for $c_r$ for which $P_1 > P_2$. In the region of $c_r$ where $P_2 > P_1$, we expect Algorithm \ref{algo_dist_0} to perform better than Algorithm \ref{algo_SOMP_det} which is observed  in the numerical results section.  Thus,  Algorithm \ref{algo_dist_0} is promising in terms of both performance and the communication overhead when $c_r$ (or $M$) is small.  However, as $c_r$ increases,  Algorithm \ref{algo_SOMP_det} performs better than  Algorithm \ref{algo_dist_0}. It is more likely that Algorithm \ref{algo_SOMP_det} estimates all the $T_0$ indices correctly compared to the event that all the  nodes estimate all the $T_0$ indices correctly in  Algorithm \ref{algo_dist_0} as $M$ increases.

\begin{algorithm}
\textbf{Inputs}:  (At the $j$-th node) $\mathbf y_j  , \mathbf B_j$, $T_0$ for $j=1,\cdots,L$\\
\textbf{Outputs}: (At the fusion center) Partial support set estimate $\hat{\mathcal U}$,  Decision statistic $\Lambda_{dist2}$, Detection decision\\
\textbf {Initialization}:\\
At the $j$-th node: $\hat{\mathcal U}_j(0) = \emptyset $, residual vector $\mathbf r_{j,0} = \mathbf y_j$ for $j=1,\cdots,L$.
\begin{enumerate}
 \item[-] \textbf {At the $j$-th node for $j=1,\cdots, L$}
 \item \emph{For} $t=1$ to $t=T_0$
\item Find the index $\beta_j(t)$ such that
$
\beta_j(t) = \underset{\omega = 1,\cdots,N}{\arg~ \max} ~ {|\langle \mathbf r_{j,t-1}, \mathbf B_j(\omega)\rangle|}
$
\item  Set $\hat{\mathcal U}_j(t) = \hat{\mathcal U}_j(t-1) \cup \{\beta_j (t)\}$
\item  Compute the projection operator $\mathbf P_{j,t} = \mathbf B_j(\hat{\mathcal U}_j(t)) \left(\mathbf B_j(\hat{\mathcal U}_j(t)) ^T \mathbf B_j(\hat{\mathcal U}_j(t)) \right)^{-1} \mathbf B_j(\hat{\mathcal U}_j(t)) ^T$. Update the residual vector:  $\mathbf r_{j,t} = (\mathbf I - \mathbf P_{j,t})\mathbf y_j$
    \item End \emph{For}
       \item  Set $\hat{\mathcal U}_j =\hat{\mathcal U}_j(T_0)$ and transmit $\mathcal U_j$ to the fusion center
\item Receive $\hat{\mathcal U}$ from the fusion center
\item Compute $\Lambda_j = || \mathbf P_{j,\hat{\mathcal U}} \mathbf y_j||_2^2 $ and transmit $\Lambda_j$ to the fusion center
\item[-] \textbf{At the fusion center}
    \item Receive $\hat{\mathcal U}_j$ for $j=1, \cdots, L$
    \item Compute $\hat{\mathcal U} = f(\{\hat{\mathcal U}_j\}_{j=1}^L )$ as discussed in subsection \ref{dist_alg1} and transmit to all the nodes
\item Receive $\Lambda_j$ for $j=1,\cdots, L$ from the nodes
\item Compute the decision statistic $\Lambda_{dist2} = \sum_{j=1}^{L} \Lambda_j$
\item Detection decision: if $\Lambda_{dist2}  \geq \tau_0$ decide $\mathcal H_1$, otherwise decide $\mathcal H_0$
 \end{enumerate}
 \caption{OMP  based sparse signal detection: distributed approach 2}\label{algo2_1}
 \end{algorithm}

\subsection{Description of Algorithm \ref{algo2_1}}\label{dist_alg1}
When estimating the partial support of size $T_0$ at each node, Algorithm \ref{algo_dist_0} does not account for the  fact that all the nodes have the same support in the support set estimation stage.  In the second distributed algorithm presented in Algorithm \ref{algo2_1}, the fact that all the nodes share the same support is taken into account by fusion as described below.
In Algorithm \ref{algo2_1}, each node estimates a partial support set of size $T_0$ after running $T_0 $ iterations of the standard OMP algorithm independently. Then, the $T_0$-length support set, denoted by $\hat{\mathcal U}_j$ is transmitted by the $j$-th node to the fusion center.   Based on $\{\hat{\mathcal U}_j\}_{j=1}^L $, the fusion center estimates a length-$k$ support set $\hat{\mathcal U}$ and transmits it back to the nodes. The $j$-th node then computes $\Lambda_j = ||\mathbf P_{j,\hat{\mathcal U}} \mathbf y_j||_2^2 $ where $ \mathbf P_{j,\hat{\mathcal U}}  = \mathbf B_j(\hat{\mathcal U}) \left(\mathbf B_j(\hat{\mathcal U}) ^T \mathbf B_j(\hat{\mathcal U}) \right)^{-1} \mathbf B_j(\hat{\mathcal U}) ^T$ and transmits back to the fusion center. The fusion center computes the decision statistic $\Lambda_{dist2} = \sum_{j=1}^L \Lambda_j $.

In this algorithm, each node communicates with the fusion center twice (steps 6 and 8). The communication overhead in step 8 is similar to that in Algorithm  \ref{algo_dist_0}.  In step 6, the estimated  partial support set $\mathcal U_j$ of size $T_0$ is transmitted to the fusion center. Then,  the fusion center computes an updated estimate for the support set which  can be larger than $T_0$ and of course less than  $k$. Let $I$ be the set containing all the indices (can have multiple occurrences of the same index) in $\{\hat{\mathcal U}_j\}_{j=1}^L $.  Let $D$ be a set  which  contains all the distinct values in $I$ ordered in a  descending manner based on the frequency of occurrence in $I$.  Let $\mathrm{length} (D) $ denote the number of elements  in $D$. If each node estimates $T_0$ correct locations of the support via OMP, $\mathrm{length} (D) $ can have maximum value of $k$. However, if there is an error in estimating the support after $T_0$ iterations at a given node, $\mathrm{length} (D) $ can be greater than $k$; in the worst case, $\mathrm{length} (D) = T_0 L$. Thus, we have, $T_0 \leq \mathrm{length} (D) \leq T_0L$. We compute $\hat{\mathcal U}$ as $\hat{\mathcal U} = D(1:\min(\mathrm{length} (D),k))$. Thus, in this algorithm, partial support estimates are fused (via the majority rule) at the fusion center to compute an  enlarged  (or sometimes  complete) support set.

When comparing Algorithms \ref{algo_dist_0} and \ref{algo2_1}, the additional communication overhead required by each node in Algorithm \ref{algo2_1} comes from  the  need for transmitting  indices of length $T_0$ in step 6. Thus, the  communication overhead in Algorithm \ref{algo2_1} can be reduced by letting all the nodes to run only $1$ iteration of the OMP algorithm; i.e.,  $T_0=1$. As illustrated in the numerical results section, Algorithm \ref{algo2_1} provides promising performance when $T_0=1$ since  the local decision statistic is computed  based on a  fused and enlarged  support set  compared to both  Algorithms \ref{algo_dist_0} and \ref{algo_SOMP_det}. However, as $T_0$ increases, it can be observed  that Algorithm \ref{algo_dist_0} performs better than Algorithm \ref{algo2_1} even though  the communication overhead of Algorithm \ref{algo2_1} increases as $T_0$ increases. Thus, the improved performance of Algorithm \ref{algo2_1} is significant with only small $T_0$ which is desirable. Further,  under certain conditions, both distributed algorithms  perform  better than the  centralized version presented in Algorithm \ref{algo_SOMP_det}. More details are provided  in the   numerical results section.

\subsection{Communication complexity}
In this subsection, we summarize the communication complexities of Algorithms \ref{algo_SOMP_det}-\ref{algo2_1} in terms of the number of messages to be transmitted by each node to the fusion center.
\begin{itemize}
\item Algorithm \ref{algo_SOMP_det}: $M$ messages per node
\item Algorithm \ref{algo_dist_0}: $1$ message per node
\item Algorithm \ref{algo2_1}: $T_0+1$ messages per node; needs feedback from the fusion center.
\end{itemize}

\section{Numerical  Results}\label{sec_simulation}
In this section, we illustrate the performance of sparse signal detection based on the knowledge of a partial support set.  The signals are assumed to be sparse in the canonical basis, so that $\mathbf B_j= \mathbf A_j$ for $j=1,\cdots,L$. The elements of $\mathbf A_j$ are drawn from a standard normal ensemble  and then $\mathbf A_j$ is  orthogonalized so that $\mathbf A_j \mathbf A_j^T = \mathbf I$. The sparse support set is selected from $[1,N]$ uniformly. The coefficients of $\mathbf s_j$ for $j=1,\cdots,L$ are taken  uniformly from $[-b,-a]\bigcup [a,b]$.
%In the following figures, we let  $a=3$ and $b=4$ unless otherwise specified.
For given $a$ and $b$,  the SNR is varied by changing $\sigma_v^2$.
The average (over all the nodes) uncompressed SNR, $\bar\gamma$,  is defined as  $\bar\gamma=\sum_{j=1}^L \frac{||\mathbf s_j||_2^2}{LN\sigma_v^2} = \sum_{j=1}^L \frac{\gamma_j}{LN}$ where $\gamma_j= \frac{||\mathbf s_j||_2^2}{\sigma_v^2}$. In the following figures, by average SNR, we mean the uncompressed SNR, $\bar\gamma$. As defined before,  $c_r = \frac{M}{N}$ is the  compression ratio at a given node.  Further, we define $ t_f = \frac{\hat t}{k}$ to be the minimum fraction of the support that needs  to be known. First, we illustrate the minimum fraction of support to be known  to achieve a desired detection performance level.
\begin{figure}
  \centering
  \centerline{\includegraphics[width=9.0cm]{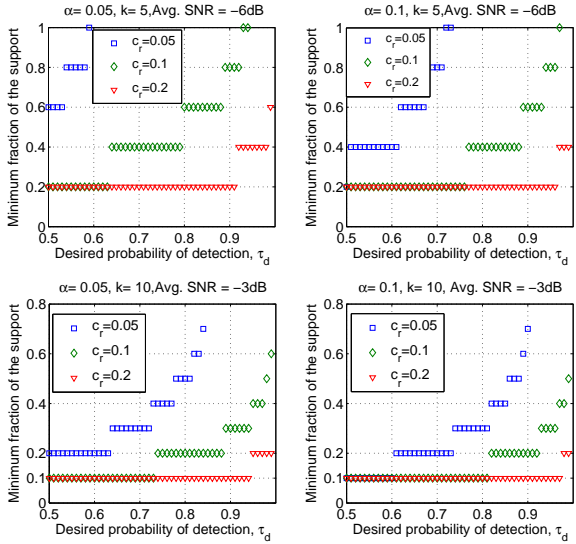}}
           \caption{Minimum fraction of the support to be known, $t_f$,  to achieve a desired detection performance, $N=256$, $L=5$, $\sigma_v^2=1$}\label{fig:min_frac1}
\end{figure}
\begin{figure}
    \centering
    \begin{subfigure}[b]{0.45\textwidth}
        \includegraphics[width=\textwidth]{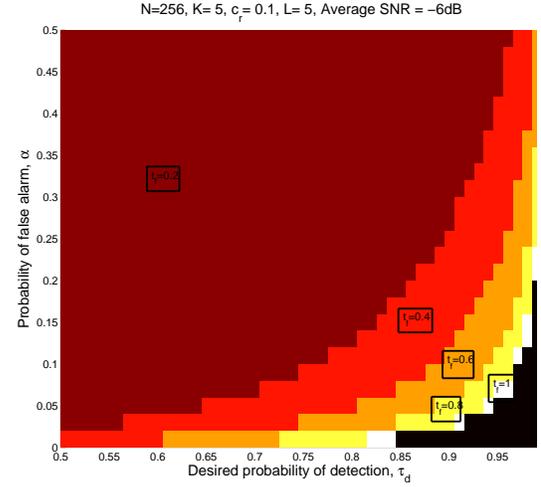}
        \caption{$L=5$, $\mathrm{SNR}=-6dB$, $\sigma_v^2=1$}
           \end{subfigure}
    \begin{subfigure}[b]{0.45\textwidth}
        \includegraphics[width=\textwidth]{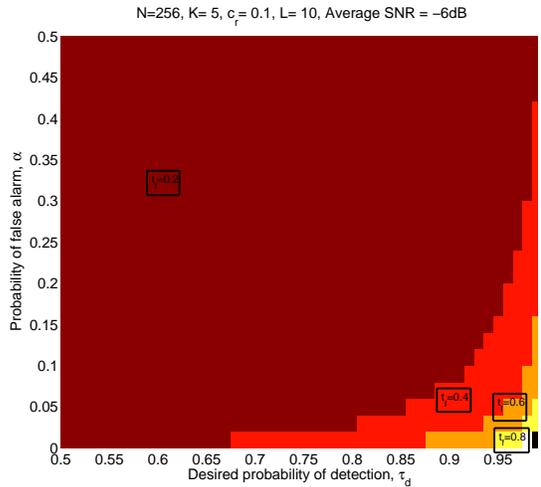}
        \caption{$L=10$, $\mathrm{SNR}=-6dB$, $\sigma_v^2=1$}
            \end{subfigure}
        \begin{subfigure}[c]{0.45\textwidth}
        \includegraphics[width=\textwidth]{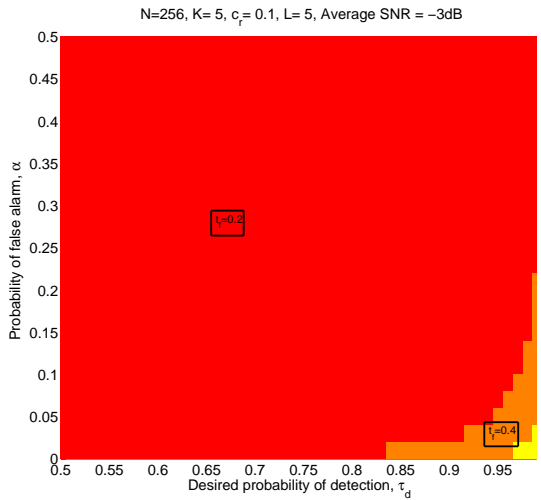}
        \caption{$L=5$, $\mathrm{SNR}=-3dB$, $\sigma_v^2=0.5$}
            \end{subfigure}
    \caption{Minimum fraction of the support to be known, $t_f$,  to achieve a desired detection performance as  $\tau_d$  and $\alpha$  vary, values of $\tau_d$ and $\alpha$ in the black region cannot be achieved even if all the indices of the support are known correctly; $N=256$, $k=5$, $c_r=0.1$ }\label{fig:min_frac2}
\end{figure}

\subsection{The minimum fraction of the support set  to achieve a desired performance level}\label{numerical_1}
  In Figs. \ref{fig:min_frac1} and \ref{fig:min_frac2},   the minimum fraction of the support set  is shown as different parameters vary. We let  $a=3$ and $b=4$. In Fig. \ref{fig:min_frac1}, $t_f$ is shown as  $\tau_d$ varies for different values of $c_r$, $\alpha$ and $k$ while keeping $N=256$, $L=5$ and $\sigma_v^2=1$.   It is worth mentioning that the average SNR in the bottom two subplots in Fig. \ref{fig:min_frac1} is different from the two top subplots due to the change of  $k$ (albeit $\sigma_v^2$ is kept constant).  When  $c_r$ increases, the fraction of the support set to be estimated  becomes small  to achieve the desired  performance level. For example, in the top left figure, to achieve $\tau_d \approx 0.9$, the knowledge of only one index of the  support set is sufficient when $c_r=0.2$ while it requires the knowledge of $4$ (out of 5) indices when $c_r=0.1$. As  $c_r$ increases, the SNR of the compressed signal increases (although the uncompressed SNR remains the same) resulting in  better performance.  Similarly, the impact of $\alpha$ and $k$ (thus the uncompressed SNR) on $t_f$  is depicted in Fig. \ref{fig:min_frac1}.

\begin{figure}
    \centering
    \begin{subfigure}[b]{0.3\textwidth}
        \includegraphics[width=\textwidth]{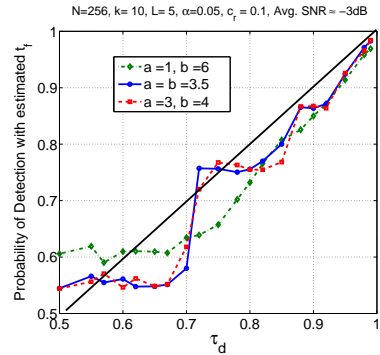}
        \caption{$N=256$, $k=10$, $\sigma_v^2=1$, $\alpha=0.05$, $c_r=0.1$}
           \end{subfigure}
           \begin{subfigure}[b]{0.3\textwidth}
        \includegraphics[width=\textwidth]{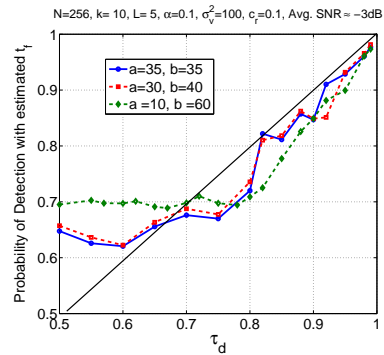}
        \caption{$N=256$, $k=10$, $\sigma_v^2=100$, $\alpha=0.1$, $c_r=0.1$}
           \end{subfigure}
         \begin{subfigure}[b]{0.3\textwidth}
        \includegraphics[width=\textwidth]{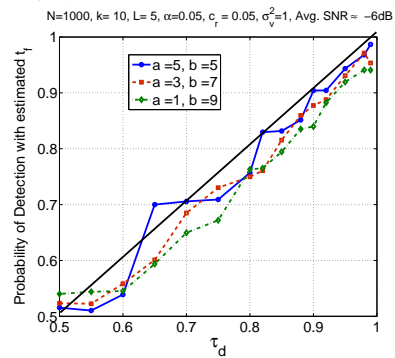}
        \caption{$N=1000$, $k=10$, $\sigma_v^2=1$, $\alpha=0.05$, $c_r=0.05$}
            \end{subfigure}
    \begin{subfigure}[b]{0.3\textwidth}
        \includegraphics[width=\textwidth]{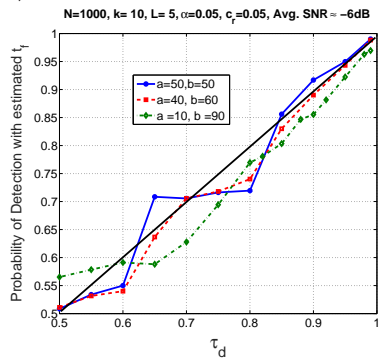}
        \caption{$N=1000$, $k=10$, $\sigma_v^2=100$, $\alpha=0.05$, $c_r=0.05$}
            \end{subfigure}
           \caption{Impact of nonzero coefficients  }\label{fig:nonzero}
\end{figure}

In resource constrained sensor  networks working in a distributed  setting, it is  desirable   to keep $c_r$ as small  as possible.
In Fig. \ref{fig:min_frac2}, the impact of $L$ and uncompressed SNR on $t_f$ is  illustrated when $c_r$ is kept at a lower value ($c_r=0.1$)  for $k=5$ and $N=256$. The regions of $\tau_d$ and $\alpha$ that can be achieved with a specific $t_f$ obtained from Proposition \ref{proposition} are shown in different colors. The black portion  corresponds to the regions of $\tau_d$ and $\alpha$ that  cannot be achieved even if all the support indices corresponding to nonzero coefficients of sparse signals are known.   It can be seen that,  when $c_r$ is small, the desired detection performance can be achieved by estimating only a very small fraction of the support set when either   $\bar\gamma$ is large or by increasing $L$. For example, as depicted in Fig. \ref{fig:min_frac2} (c) when $L=5$ and uncompressed $\mathrm{SNR}=-3dB$, almost all the regions of $\tau_d$ and $\alpha$ can be achieved by  knowing only  one or two  indices of  the  support set correctly (so that $t_f=0.2$ or $t_f=0.4$). In summary, Figures \ref{fig:min_frac1} and  \ref{fig:min_frac2} demonstrate  the regions of $k$, $L$, $\bar\gamma$, $\tau_d$, $\alpha$ and $c_r$  so that estimation of only a small fraction of the support is sufficient for sparse signal detection with  desired performance.
\begin{figure}
    \centering
    \begin{subfigure}[b]{0.5\textwidth}
        \includegraphics[width=\textwidth]{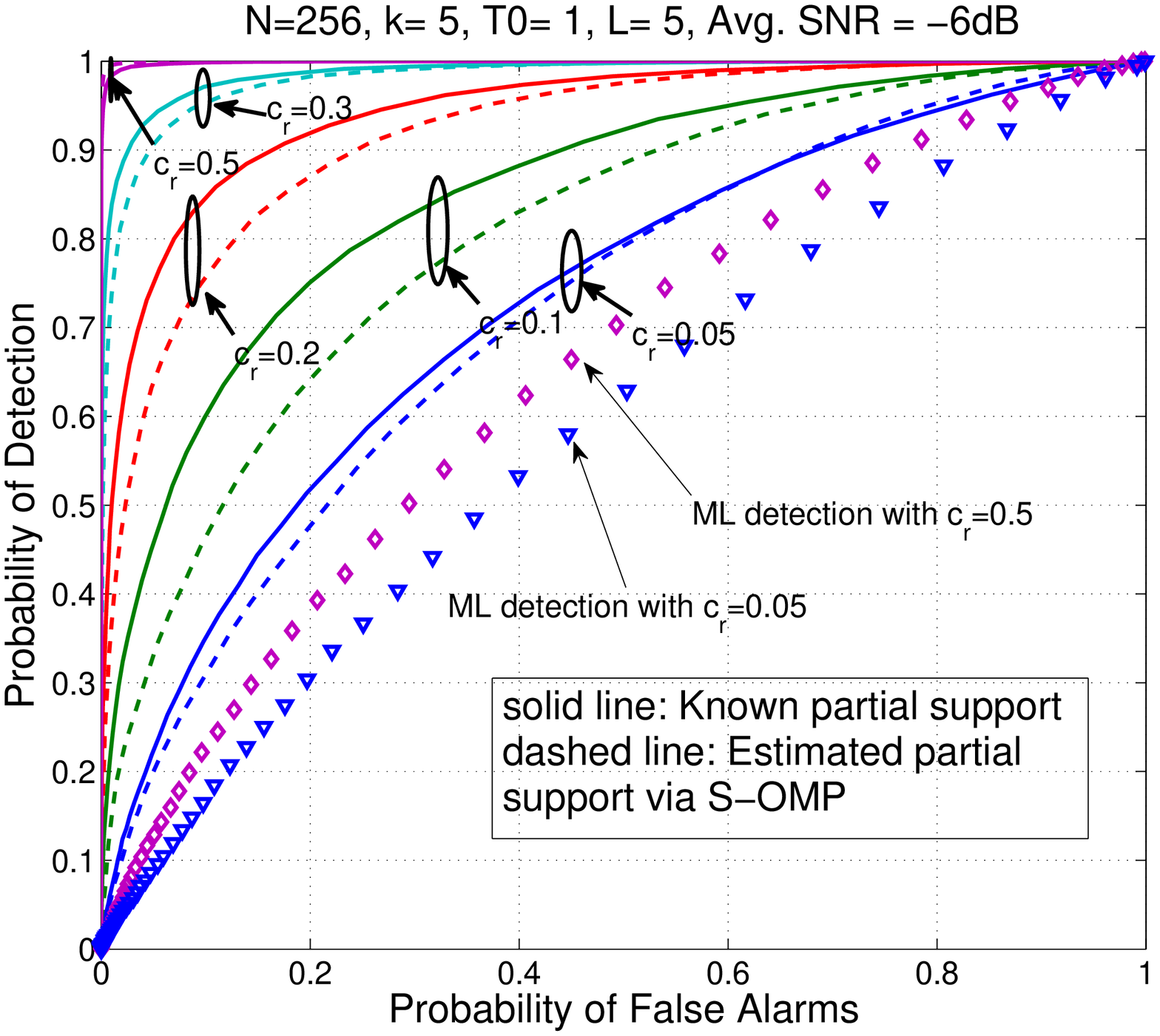}
        \caption{$T_0=1$}
        \label{fig:gull}
    \end{subfigure}
    ~ %add desired spacing between images, e. g. ~, \quad, \qquad, \hfill etc.
    \begin{subfigure}[b]{0.5\textwidth}
        \includegraphics[width=\textwidth]{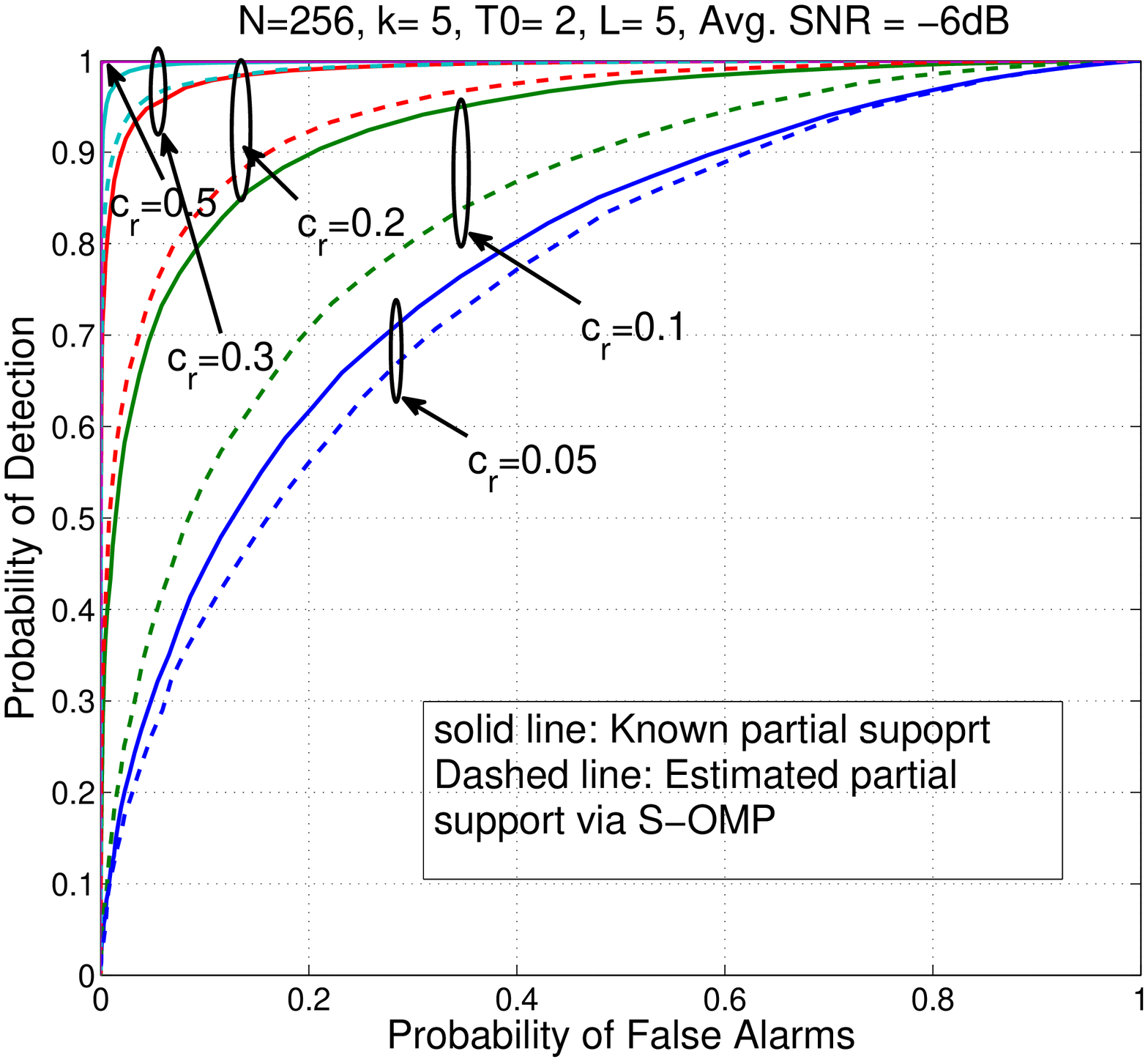}
        \caption{$T_0=2$}
        \label{fig:tiger}
    \end{subfigure}
        ~ %add desired spacing between images, e. g. ~, \quad, \qquad, \hfill etc.
    \caption{ROC curves for detection based on partial support set estimation, $k=5$, $N=256$, $L=5$}\label{fig:ROC_partial}
\end{figure}

\subsection{Impact of the values of nonzero coefficients}
{The minimum fraction of nodes in Proposition  \ref{proposition} was computed under the assumption that the nonzero coefficients of the sparse signals do not deviate much from each other. In this subsection, we analyze the results in  Proposition  \ref{proposition} when this assumption is relaxed.   In Fig. \ref{fig:nonzero}, we plot the performance of  the detector (\ref{Lamda_1}) when the  size of the support is  computed  based on  Proposition \ref{proposition} given that  the  nonzero coefficients are actually far from each other.  For a given  $\tau_d$ on the x-axis, $t_f$ is computed as in ({\ref{hat_u}}). Substituting  $T_0 = t_f k$ in (\ref{Lamda_1}), the probability of detection (obtained numerically)  with the  decision statistic (\ref{Lamda_1}) is shown on the y-axis. It is noted that since Proposition \ref{proposition} does not give a clue on which $T_0$ indices should be selected, we get $T_0$ indices uniformly from $\mathcal U$. The first two subplots correspond to a small problem size  while last two subplots consider relatively large $N$. For both problem sizes, relatively large and small values for $a$ and $b$ are considered. It is noted  that  $a$, $b$    and $\sigma_v^2$ are  changed accordingly  so that the average SNR remains approximately the same for given $N$ and $L$.  The desired scenario is that the  curves stay close or above the black (y=x) line. It can be seen that when $N$ is large,  all the curves remain fairly close  to (or above)  the black line for most values of   $\tau_d$. When $N$ is not very large (subplots (a) and (b)), performance degradation can be seen for  some  values of  $\tau_d$.  However, for relatively  large values of  $\tau_d$  (which is the most interesting scenario),   there is no significant performance degradation using Proposition \ref{proposition} even the coefficient  values  deviate quite significantly from each other irrespective of the problem size.}

\subsection{Performance of sparse signal detection in a centralized setting: Detection with known partial support and via S-OMP}
In this subsection, we compare the detection performance when the partial support set is exactly known and it is estimated via S-OMP in a centralized setting.
In Fig. \ref{fig:ROC_partial}, we show receiver operating characteristic (ROC) curves with $10^4$ Monte carlo runs. In the case where the partial support set  is known,  $T_0$ indices  are selected randomly from $\mathcal U$ and the results are averaged over $20$ trials.
In the S-OMP based detection, $T_0$ indices are  selected according to  Algorithm \ref{algo_SOMP_det}.
In Fig. \ref{fig:ROC_partial} (a), $T_0=1$ while in Fig. \ref{fig:ROC_partial} (b),  $T_0=2$. The other parameters $k$, $L$, and $N$  remain the same in both subfigures  as stated while $a=3$ and $b=4$. In Fig. \ref{fig:ROC_partial} (a), we further plot the detection performance when the sparse signal is estimated via  maximum likelihood (ML) estimation ignoring sparsity; i.e. when the decision statistic is $\Lambda_{ML} =  \sum_{j=1}^L || \mathbf P_{j}\mathbf y_j||_2^2$ with $\mathbf P_{j} = \mathbf B_j\left(\mathbf B_j ^T \mathbf B_j\right)^{-1} \mathbf B_j ^T$. It is observed that exploitation of sparsity even with $T_0=1$ outperforms the ML based detection approach and we avoid plotting curves for  this comparison in subsequent figures.

It is observed that, for both  $T_0=1$ and $T_0=2$, the performance of the S-OMP based detection is close to the performance with known  partial support  for very small and relatively large $c_r$ values.  When $c_r$ is very small (e.g $\approx 0.05$ in Fig.\ref{fig:ROC_partial}), the SNR of the compressed observation vector is small, thus even with the known support set, the power of the compressed   observations projected onto the known subspace is not significant compared to the analogous  noise power. Thus, close (and poor) performance  is observed when the partial support is known or estimated. On the other hand, when $c_r$ is large, the  estimated support set of size  $T_0$ via Algorithm \ref{algo_SOMP_det} is more accurate, resulting in close performance to the case where the partial support  is exactly known.  However, when $c_r$ is moderate (i.e. $\approx .1$ or $.2$ in Fig. \ref{fig:ROC_partial})), S-OMP based detection has a performance gap compared to detection with known partial support of the same size. When $c_r$ takes such values, the accuracy of the estimated partial support set via S-OMP  is not quite good thus, resulting in a performance gap. In such regions of $c_r$ where  compressed measurements per node  are not sufficient to provide a good estimate of the support set via S-OMP, the detection performance  is improved with the two distributed algorithms as will be illustrated next.

\begin{figure*}
    \centering
    \begin{subfigure}[b]{0.38\textwidth}
        \includegraphics[width=\textwidth]{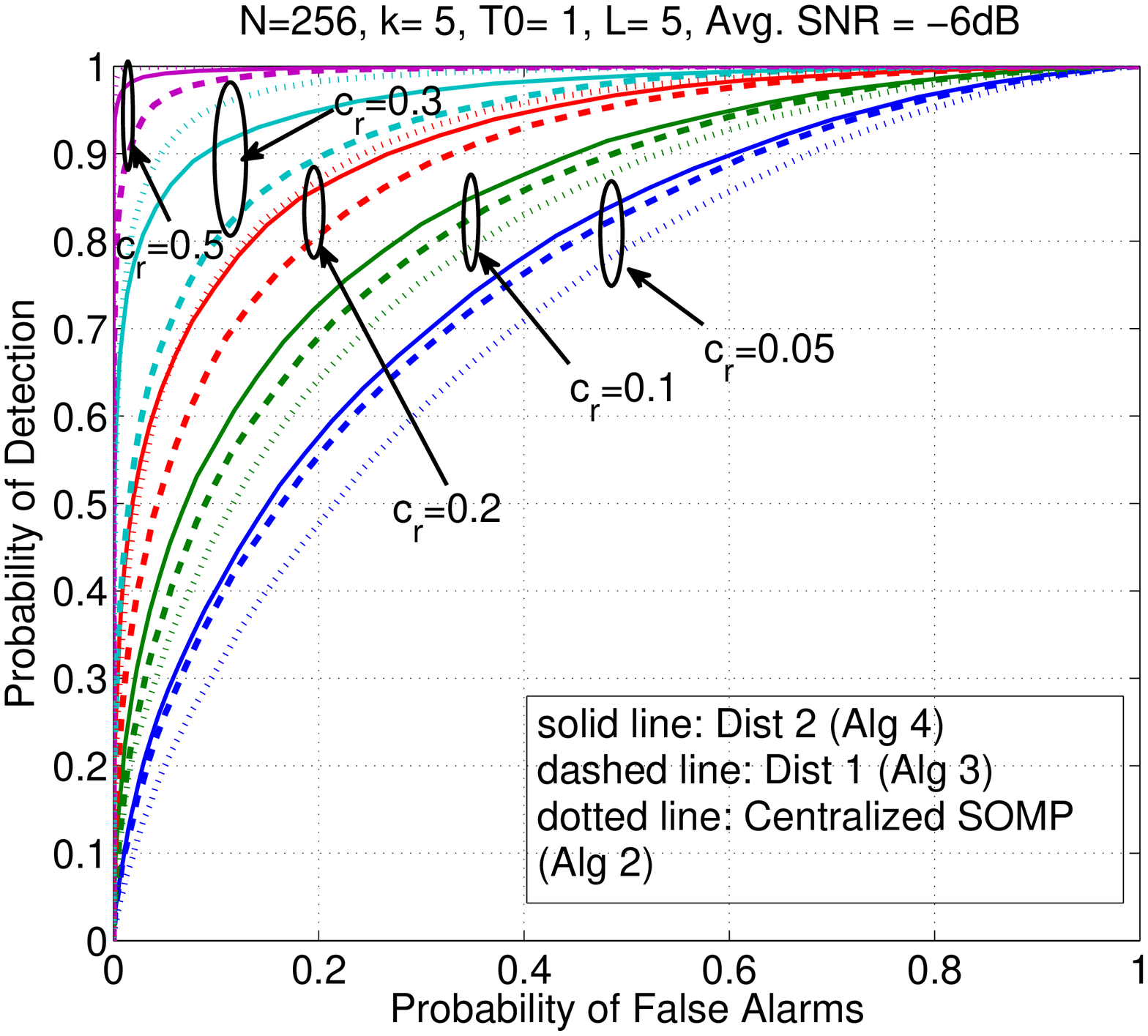}
        \caption{$N=256$, $k=5$, Average $\mathrm{SNR}=-6dB$}
        \label{fig:gull}
    \end{subfigure}
    ~ %add desired spacing between images, e. g. ~, \quad, \qquad, \hfill etc.
    \begin{subfigure}[b]{0.38\textwidth}
        \includegraphics[width=\textwidth]{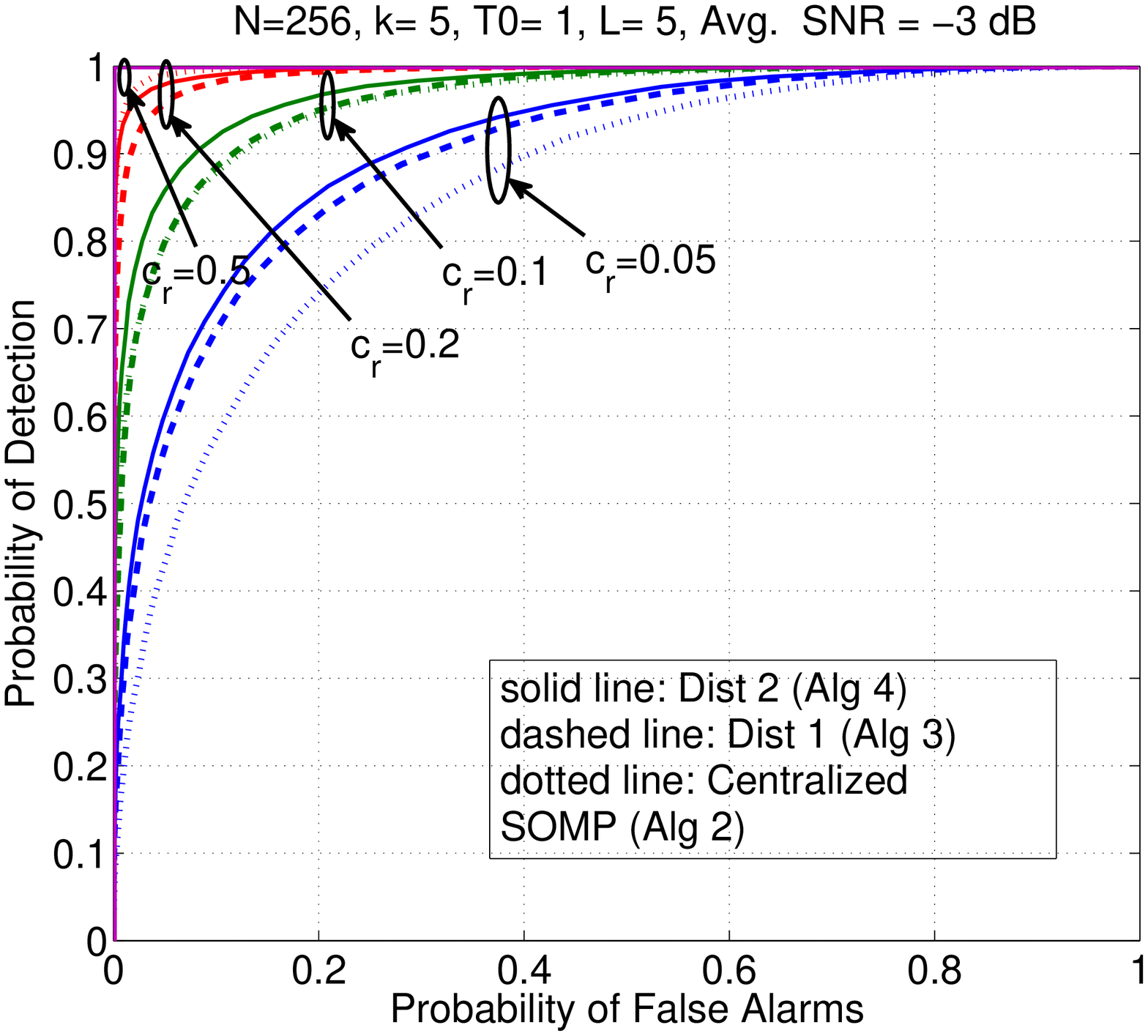}
        \caption{$N=256$, $k=5$, Average $\mathrm{SNR}=-3dB$}
        \label{fig:tiger}
    \end{subfigure}
        \begin{subfigure}[b]{0.38\textwidth}
        \includegraphics[width=\textwidth]{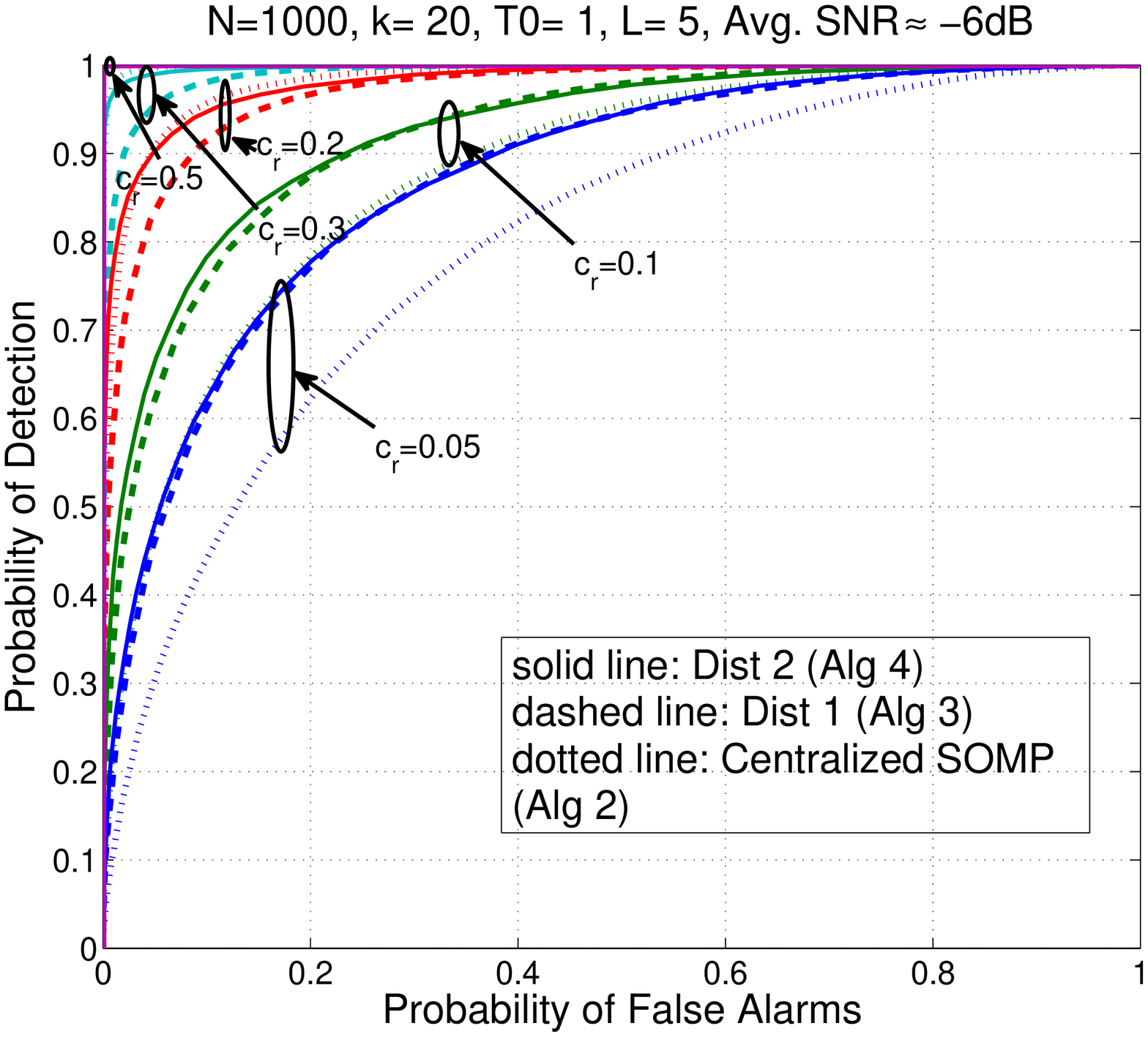}
        \caption{$N=1000$, $k=20$, $T_0=1$}
        \label{fig:gull}
    \end{subfigure}
       \begin{subfigure}[b]{0.38\textwidth}
        \includegraphics[width=\textwidth]{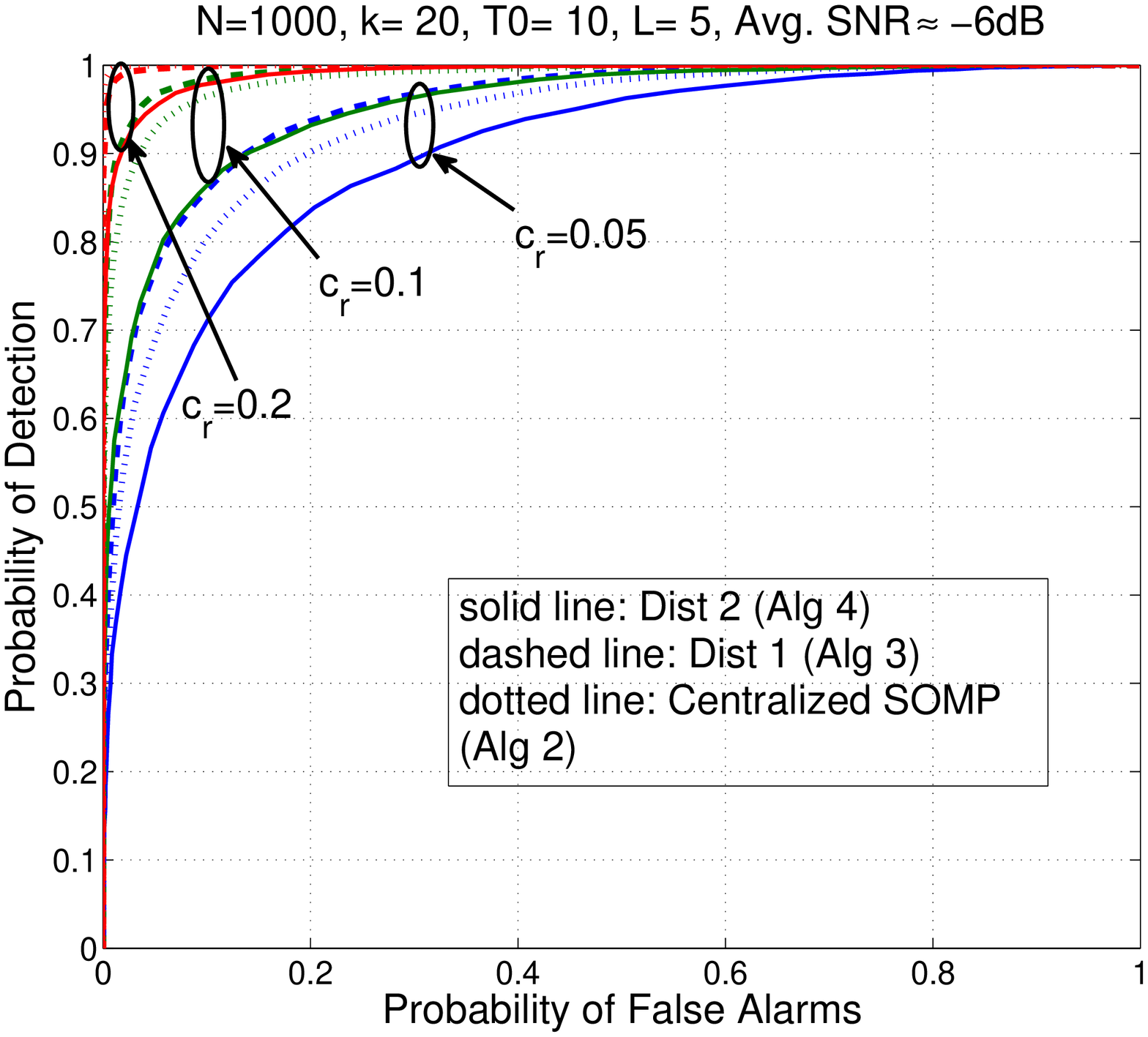}
        \caption{$N=1000$, $k=20$, $T_0=10$}
        \label{fig:tiger}
    \end{subfigure}
    ~ %add desired spacing between images, e. g. ~, \quad, \qquad, \hfill etc.
    \caption{ROC curves for OMP based detection; S-OMP and two distributed approaches  }\label{fig:ROC_OMP}
\end{figure*}
\subsection{Performance comparison  of S-OMP based and two distributed OMP based algorithms}
In Fig. \ref{fig:ROC_OMP},  ROC curves are plotted  for different values of $c_r$  for centralized  and two distributed OMP based algorithms. We let $a=3$ and $b=4$.  In Figs. \ref{fig:ROC_OMP} (a) and (b), we consider relatively a small sized problem with  $N=256$, $k=5$, $L=5$, $T_0=1$ and the average SNR (uncompressed) is varied by changing  $\sigma_v^2$. On the other hand, in Figs. \ref{fig:ROC_OMP} (c) and (d), a bigger sized problem is considered so that $N=1000$, $k=20$, $L=5$, the average SNR (uncompressed) $=-6~dB$ and $T_0$ is varied.   We make several important observations here.
\begin{enumerate}
\item When $c_r$ and $T_0$ are  quite  small, the S-OMP based algorithm performs worse than the two distributed algorithms (subplots (a), (b) and (c)).
    This (counter-intuitive) phenomenon was discussed in detail  in Subsection \ref{compare_2_algo} for $T_0=1$ considering  Algorithms \ref{algo_SOMP_det} and   \ref{algo_dist_0}.  Thus, Algorithm  \ref{algo_dist_0}  produces  a better decision  statistic to discriminate between the noise and the signal.
    Similar explanation holds for Algorithm \ref{algo2_1}  which provides even better results compared to Algorithm \ref{algo_dist_0} due to the fusion  of support indices estimated at multiple nodes.   Thus, when $M$ (thus $c_r$) is not sufficient to provide accurate estimates for the support set after $T_0$ iterations via S-OMP,  the two distributed algorithms, by fusion, provide better performance.
     %However, as $c_r$, the S-OMP based detector performs better than the two distributed algorithms for same $T_0$.
     \item     When $T_0$ is relatively large, from Fig. \ref{fig:ROC_OMP}  (d), it is seen that Algorithm \ref{algo2_1} works relatively poorly  compared to the other  two algorithms. As discussed in Section \ref{dist_alg1} and seen in Fig. \ref{fig:ROC_OMP} (a), (b) and (c), Algorithm \ref{algo2_1} is promising in terms of communication complexity only when $T_0$ is  small.
    \end{enumerate}

 It is noted that, in the S-OMP algorithm, raw compressed observations are fused in computing the support set indices and the decision statistic for detection is obtained based on such estimated indices. On the  other  hand, in Algorithms \ref{algo_dist_0}  and \ref{algo2_1}, individually computed local decision statistics are fused to compute the final decision statistic at the fusion center.
 From Fig. \ref{fig:ROC_OMP}, it is seen that, measurement fusion via S-OMP provides better performance only when $c_r$ is relatively large.  This is an important  observation which is somewhat counter intuitive, since one would expect a centralized  scheme to work better  than any distributed  approach. There are several reasons. One is that S-OMP is not an optimal method of computing the sparse support set although it  provides promising results when $c_r$ exceeds a certain threshold. When $c_r$ is small, there can be other  variants (such as two distributed algorithms presented here) of OMP other than S-OMP that would provide better  performance in sparse support recovery. Another reason is that, we focus on partial signal recovery (and detection based on that) but not on complete  signal recovery. Thus, we conclude that better decision statistics  for detection based on OMP  can be computed in a distributed setting compared to a centralized setting  under certain conditions.

\begin{figure}
    \centering
    \begin{subfigure}[b]{0.32\textwidth}
        \includegraphics[width=\textwidth]{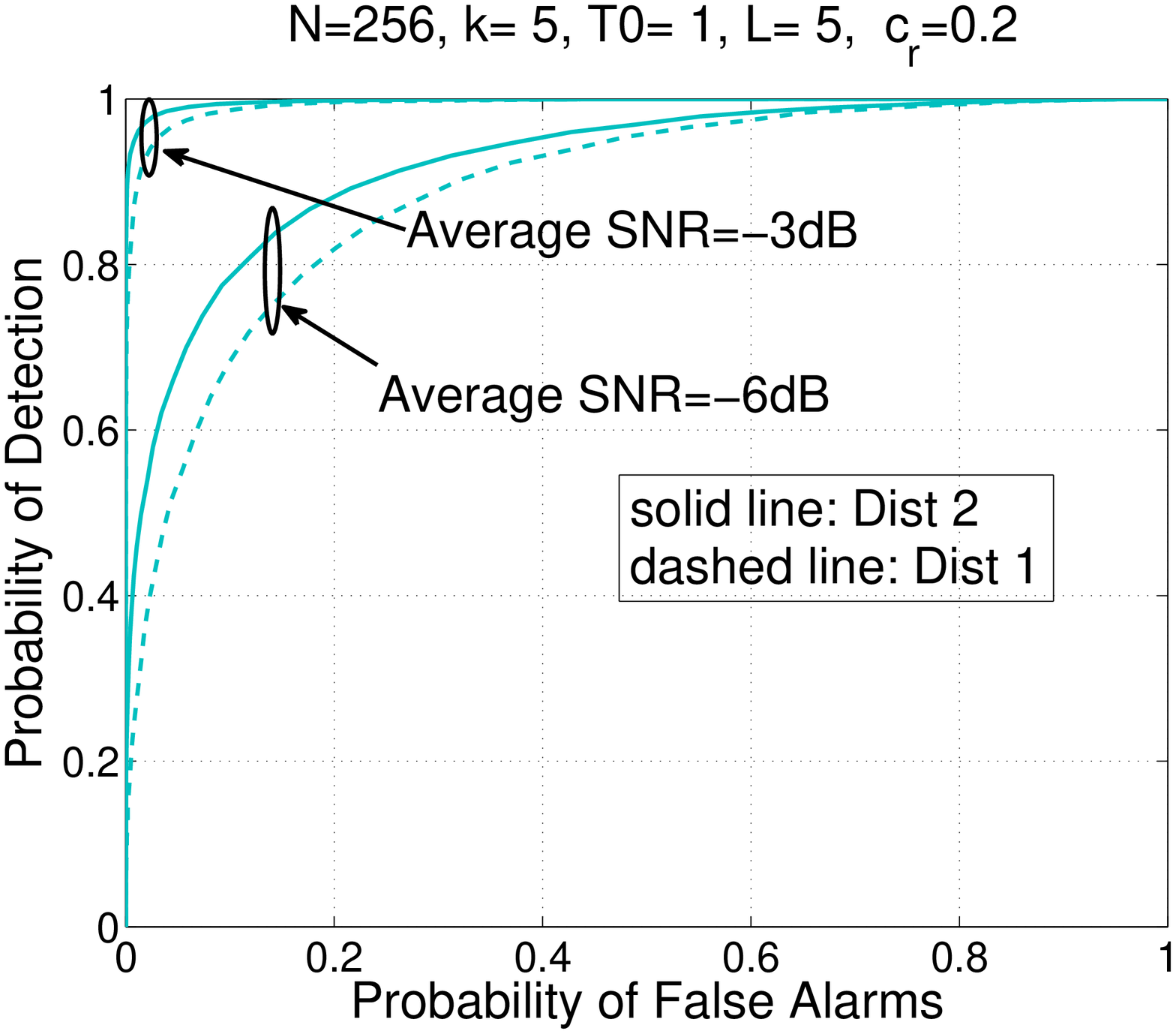}
        \caption{ $T_0=1$, $L=5$}
        \label{fig:gull}
    \end{subfigure}
         \begin{subfigure}[b]{0.32\textwidth}
        \includegraphics[width=\textwidth]{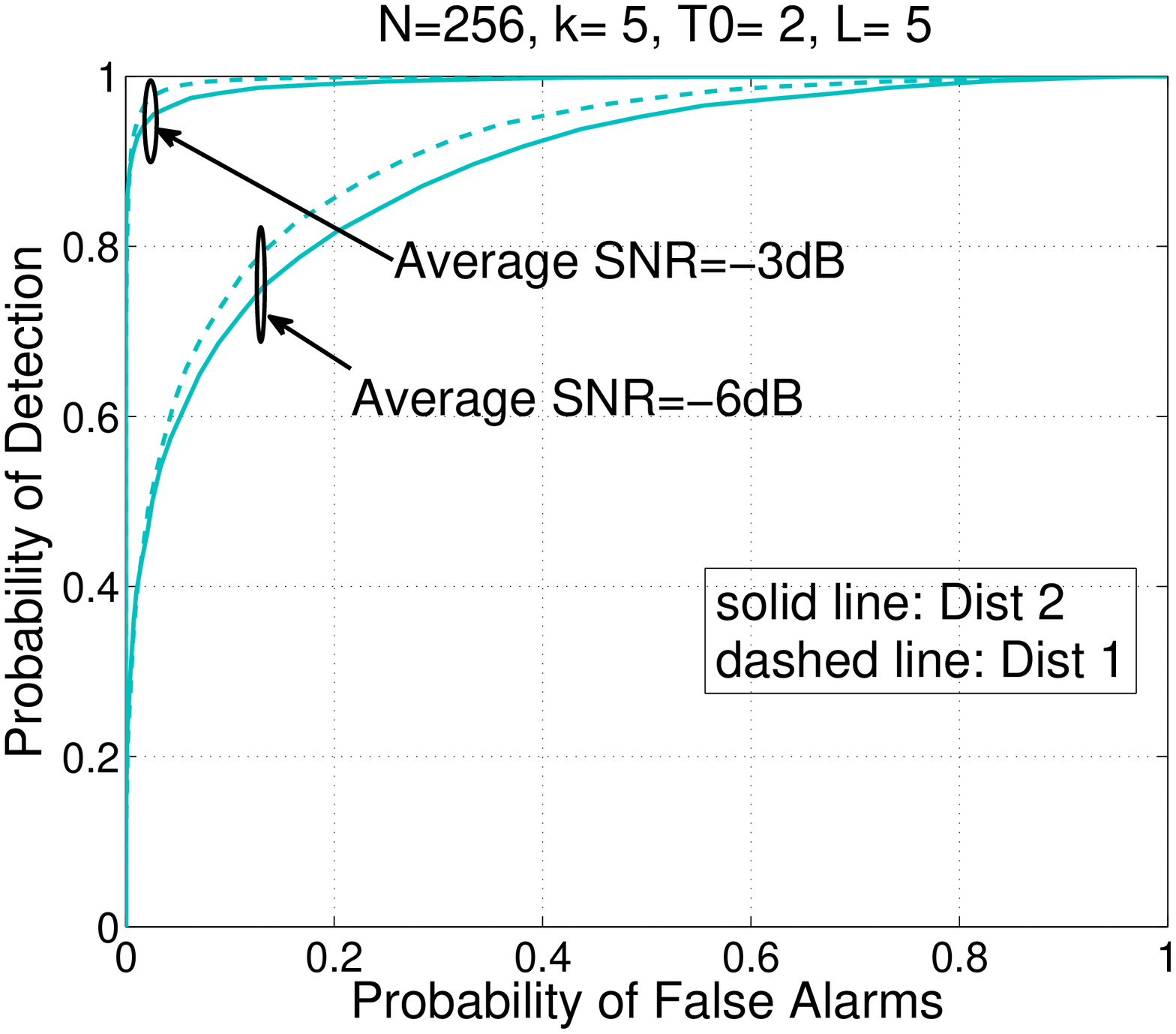}
        \caption{ $T_0=2$, $L=5$}
        \label{fig:tiger}
    \end{subfigure}
 \begin{subfigure}[b]{0.32\textwidth}
        \includegraphics[width=\textwidth]{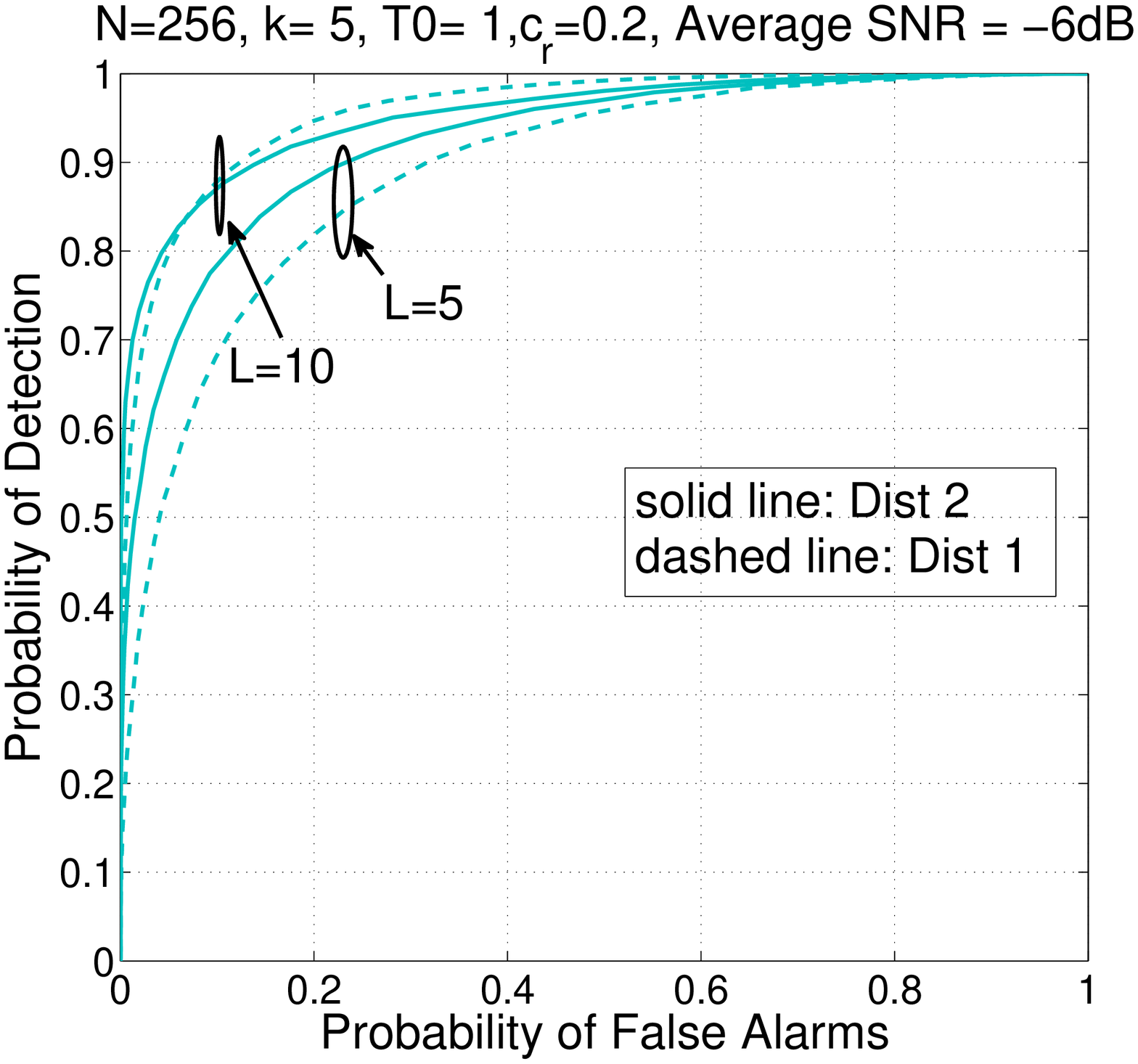}
        \caption{ $T_0=1$, $\mathrm{SNR}=-6dB$}
        \label{fig:tiger}
    \end{subfigure}
    \caption{Performance comparison of two distributed algorithms, $k=5$, $N=256$, $c_r=0.2$}\label{fig:ROC_Dist1_2}
\end{figure}

\subsection{Performance comparison  of two distributed OMP based algorithms}
The first distributed algorithm presented  in Algorithm \ref{algo_dist_0}  requires small  communication overhead compared to the second distributed  algorithm in Algorithm \ref{algo2_1}. In Fig. \ref{fig:ROC_Dist1_2},  we compare the performance of the two distributed algorithms as   $T_0$, SNR values  and $L$ vary. When $T_0=1$, it can be seen that Algorithm \ref{algo2_1} performs better  than Algorithm \ref{algo_dist_0} for both SNR values considered  when $L=5$. However, as seen in Fig. \ref{fig:ROC_Dist1_2}(b), when $T_0$ is increased,  Algorithm \ref{algo_dist_0} performs better  than Algorithm \ref{algo2_1}. The additional communication overhead required by Algorithm \ref{algo2_1}  is not worth it as  $T_0$ increases compared  to Algorithm \ref{algo_dist_0}. Thus, it is seen that the  Algorithm \ref{algo2_1} is promising and worth the additional communication overhead needed  only when  $T_0=1$ and $L$ is relatively small, which are the most desirable scenarios.

%%%%%%%%%%%%%%%%%%%%%%%%%%%%%%%%%%%%%%%%%%%%%%%%%%%%%%%%%%%%%%%%%%%%%%%%%%%

\section{Conclusion}\label{sec_conclusion}
In this paper, we discussed the use of a  CS measurement scheme for sparse signal detection  when multiple sparse signals observed by distributed nodes  share the same sparsity pattern. We showed that by estimating only a fraction of the support with less computational power  than that is required for complete signal recovery, a reliable decision statistic can be designed. First, we analyzed the minimum fraction of the support set  to be estimated to achieve a desired detection performance in a centralized setting. Then, OMP based algorithms were developed to jointly  estimate a partial support set and perform detection in centralized as well as distributed settings. It was shown that with each node  estimating only one index of the support set, a reliable detection decision  can be made by appropriate fusion among nodes. Further, when the number of compressed measurements acquired at each node is small, the two distributed algorithms (with less communication overhead) are shown to outperform the centralized algorithm (with higher communication overhead). {In future work, we will show the efficiency and effectiveness of the proposed algorithms with different real world application scenarios (using real experimental data).}

%\newpage
\bibliographystyle{IEEEtran}
\bibliography{IEEEabrv,bib1}

\begin{thebibliography}{10}
\providecommand{\url}[1]{#1}
\csname url@samestyle\endcsname
\providecommand{\newblock}{\relax}
\providecommand{\bibinfo}[2]{#2}
\providecommand{\BIBentrySTDinterwordspacing}{\spaceskip=0pt\relax}
\providecommand{\BIBentryALTinterwordstretchfactor}{4}
\providecommand{\BIBentryALTinterwordspacing}{\spaceskip=\fontdimen2\font plus
\BIBentryALTinterwordstretchfactor\fontdimen3\font minus
  \fontdimen4\font\relax}
\providecommand{\BIBforeignlanguage}[2]{{%
\expandafter\ifx\csname l@#1\endcsname\relax
\typeout{** WARNING: IEEEtran.bst: No hyphenation pattern has been}%
\typeout{** loaded for the language `#1'. Using the pattern for}%
\typeout{** the default language instead.}%
\else
\language=\csname l@#1\endcsname
\fi
#2}}
\providecommand{\BIBdecl}{\relax}
\BIBdecl

\bibitem{Davenport_book2012}
M.~Davenport, M.~Duarte, Y.~Eldar, and G.~Kutyniok, \emph{Introduction to
  compressed sensing; Chapter in Compressed Sensing: Theory and
  Applications}.\hskip 1em plus 0.5em minus 0.4em\relax Cambridge University
  Press, 2012.

\bibitem{candes1}
E.~Cand$\grave{e}$s, J.~Romberg, and T.~Tao, ``Robust uncertainty principles:
  exact signal reconstruction from highly incomplete frequency information,''
  \emph{IEEE Trans. Inf. Theory}, vol.~52, no.~2, pp. 489 -- 509, Feb. 2006.

\bibitem{candes2}
E.~Cand$\grave{e}$s and T.~Tao, ``Near-optimal signal recovery from random
  projections: Universal encoding strategies?'' \emph{IEEE Trans. Inf. Theory},
  vol.~52, no.~12, pp. 5406 -- 5425, Dec. 2006.

\bibitem{donoho1}
D.~Donoho, ``Compressed sensing,'' \emph{IEEE Trans. Inf. Theory}, vol.~52,
  no.~4, pp. 1289--1306, Apr. 2006.

\bibitem{candes_TIT1}
E.~J. Cand\`{e}s and Y.~Plan, ``A probabilistic and \textsc{RIP}less theory of
  compressed sensing,'' \emph{IEEE Trans. Inf. Theory}, vol.~57, no.~11, pp.
  7235--7254, 2011.

\bibitem{Eldar_B1}
Y.~C. Eldar and G.~Kutyniok, \emph{Compressed Sensing: Theory and
  Applications}.\hskip 1em plus 0.5em minus 0.4em\relax Cambridge University
  Press, 2012.

\bibitem{Baraniuk4}
R.~G. Baraniuk, V.~Cevher, M.~Duarte, and C.Hegde, ``Model based compressed
  sensing,'' \emph{IEEE Trans. Inf. Theory}, vol.~56, no.~4, pp. 1982--2001,
  Apr. 2010.

\bibitem{duarte_ICASSP06}
M.~F. Duarte, M.~A. Davenport, M.~B. Wakin, and R.~G. Baraniuk, ``Sparse signal
  detection from incoherent projections,'' in \emph{Proc. Acoust., Speech,
  Signal Processing (ICASSP)}, May 2006.

\bibitem{haupt_ICASSP07}
J.~Haupt and R.~Nowak, ``Compressive sampling for signal detection,'' in
  \emph{Proc. Acoust., Speech, Signal Processing (ICASSP)}, vol.~3, Honolulu,
  Hawaii, Apr. 2007, pp. III--1509 -- III--1512.

\bibitem{Wang_ICASSP08}
Z.~Wang, G.~Arce, and B.~Sadler, ``Subspace compressive detection for sparse
  signals,'' in \emph{Proc. Acoust., Speech, Signal Processing (ICASSP)}, Mar.
  2008, pp. 3873--3876.

\bibitem{meng_CISS09}
J.~Meng, H.~Li, and Z.~Han, ``Sparse event detection in wireless sensor
  networks using compressive sensing,'' in \emph{43rd Annual Conf. on
  Information Sciences and Systems (CISS)}, Baltimore, MD, Mar. 2009, pp. 181
  -- 185.

\bibitem{davenport_JSTSP10}
M.~A. Davenport, P.~T. Boufounos, M.~B. Wakin, and R.~Baraniuk, ``Signal
  processing with compressive measurements,'' \emph{IEEE J. Sel. Topics Signal
  Process.}, vol.~4, no.~2, pp. 445 -- 460, Apr. 2010.

\bibitem{Wimalajeewa_asilomar10}
T.~Wimalajeewa, H.~Chen, and P.~K. Varshney, ``Performance analysis of
  stochastic signal detection with compressive measurements,'' in
  \emph{$44^{\mathrm{th}}$ Annual Asilomar Conf. on Signals, Systems and
  Computers}, Nov. 2010, pp. 913--817.

\bibitem{Zahedi_Phys12}
R.~Zahedi, A.~Pezeshki, and E.~K. Chong, ``Measurement design for detecting
  sparse signals,'' \emph{Physical Communication, Compressive Sensing in
  Communications}, vol.~5, no.~2, pp. 64--75, 2012.

\bibitem{Wimalajeewa_ICASSP13}
T.~Wimalajeewa and P.~K. Varshney, ``Cooperative sparsity pattern recovery in
  distributed networks via distributed-\textsc{OMP},'' in \emph{Proc. Acoust.,
  Speech, Signal Processing (ICASSP)}, Vancouver, BC, May 2013, pp. 5288--5292.

\bibitem{Gang_globalsip14}
G.~Li, H.~Zhang, T.~Wimalajeewa, and P.~K. Varshney, ``On the detection of
  sparse signals with sensor networks based on \textsc{S}ubspace
  \textsc{P}ursuit,'' in \emph{IEEE Global Conference on Signal and Information
  Processing (GlobalSIP)}, Atlanta, GA, Dec. 2014, pp. 438--442.

\bibitem{Bhavya_cscps14}
B.~Kailkhura, T.~Wimalajeewa, L.~Shen, and P.~K. Varshney, ``Distributed
  compressive detection with perfect secrecy,'' in \emph{2nd Int. Workshop on
  Compressive Sensing in Cyber-Physical Systems (CSCPS'14)}, Oct. 2014.

\bibitem{Bhavya_asilomar14}
B.~Kailkhura, T.~Wimalajeewa, and P.~K. Varshney, ``On physical layer secrecy
  of collaborative compressive detection,'' in \emph{$48^{\mathrm{th}}$ Annual
  Asilomar Conf. on Signals, Systems and Computers}, 2014.

\bibitem{Zheng_icc11}
H.~Zheng, S.~Xiao, and X.~Wang, ``Sequential compressive target detection in
  wireless sensor networks,'' in \emph{IEEE Int. Conf. on Communications
  (ICC)}, Kyoto,, June 2011, pp. 1 --5.

\bibitem{Rao_icassp2012}
B.~S. M.~R. Rao, S.~Chatterjee, and B.~Ottersten, ``Detection of sparse random
  signals using compressive measurements,'' in \emph{Proc. Acoust., Speech,
  Signal Processing (ICASSP)}, 2012, pp. 3257--3260.

\bibitem{Cao_Info2014}
J.~Cao and Z.~Lin, ``Bayesian signal detection with compressed measurements,''
  \emph{Information Sciences}, pp. 241--253, 2014.

\bibitem{Kailkhura_WCL16}
B.~Kailkhura, S.~Liu, T.~Wimalajeewa, and P.~K. Varshney, ``Measurement matrix
  design for compressed detection with secrecy guarantees,'' \emph{IEEE
  Wireless Commun. Lett.}, 2016, Accepted.

\bibitem{Kailkhura_TSP16}
B.~Kailkhura, T.~Wimalajeewa, and P.~K. Varshney, ``Collaborative compressive
  detection with physical layer secrecy constraints,'' \emph{IEEE Trans. Signal
  Process.}, 2016, Submitted.

\bibitem{Reeves_ISIT08}
G.~Reeves and M.~Gastpar, ``Sampling bounds for sparse support recovery in the
  presence of noise,'' in \emph{IEEE Int. Symp. on Information Theory (ISIT)},
  Toronto, ON, Jul. 2008, pp. 2187--2191.

\bibitem{Tropp_P12006}
J.~Tropp, A.~Gilbert, and M.~Strauss, ``Algorithms for simultaneous sparse
  approximation. part \textsc{I}: Greedy pursuit,'' \emph{Signal Processing,
  special issue on Sparse approximations in signal and image processing},
  vol.~86, no.~4, pp. 572--588, 2006.

\bibitem{Tropp_P22006}
------, ``Algorithms for simultaneous sparse approximation. part \textsc{II}:
  Convex relaxation,'' \emph{Signal Processing, special issue on Sparse
  approximations in signal and image processing}, vol.~86, no.~4, pp. 589--602,
  2006.

\bibitem{Cotter1}
S.~F. Cotter, B.~D. Rao, K.~Engan, and K.~Kreutz-Delgado, ``Sparse solutions to
  linear inverse problems with multiple measurement vectors,'' \emph{IEEE
  Trans. Signal Process.}, vol.~53, no.~7, pp. 2477--2488, July 2005.

\bibitem{Chen2}
J.~Chen and X.~Huo, ``Theoretical results on sparse representations of
  multiple-measurement vectors,'' \emph{IEEE Trans. Signal Process.}, vol.~54,
  no.~12, pp. 4634--4643, Dec. 2006.

\bibitem{Obozinski2}
G.~Obozinski, M.Wainwright, and M.~Jordan, ``Support union recovery in
  high-dimensional multivariate regression,'' \emph{Ann. Stat.}, vol.~39,
  no.~1, pp. 1--47, 2011.

\bibitem{Wipf2}
D.~Wipf and B.~Rao, ``An empirical bayesian strategy for solving the
  simultaneous sparse approximation problem,'' \emph{IEEE Trans. Signal
  Process.}, vol.~55, no.~7, pp. 3704--3716, July 2007.

\bibitem{Eldar4}
Y.~C. Eldar and H.~Rauhut, ``Average case analysis of multichannel sparse
  recovery using convex relaxation,'' \emph{IEEE Trans. Inf. Theory}, vol.~56,
  no.~1, pp. 505--519, Jan. 2010.

\bibitem{Eldar1}
Y.~C. Eldar and M.~Mishali, ``Robust recovery of signals from a structured
  union of subspaces,'' \emph{IEEE Trans. Inf. Theory}, vol.~55, no.~11, pp.
  5302--5316, Nov. 2009.

\bibitem{ling1}
Q.~Ling and Z.~Tian, ``Decentralized support detection of multiple measurement
  vectors with joint sparsity,'' in \emph{Proc. Acoust., Speech, Signal
  Processing (ICASSP)}, 2011, pp. 2996--2999.

\bibitem{Zeng1}
F.~Zeng, C.~Li, and Z.~Tian, ``Distributed compressive spectrum sensing in
  cooperative multihop cognitive networks,'' \emph{IEEE J. Sel. Topics Signal
  Process.}, vol.~5, no.~1, pp. 37--48, Feb. 2011.

\bibitem{Bazerque1}
J.~A. Bazerque and G.~B. Giannakis, ``Distributed spectrum sensing for
  cognitive radio networks by exploiting sparsity,'' \emph{IEEE Trans. Signal
  Process.}, vol.~58, no.~3, pp. 1847--1862, Mar. 2010.

\bibitem{Ling2}
Q.~Ling and Z.~Tian, ``Decentralized sparse signal recovery for compressive
  sleeping wireless sensor networks,'' \emph{IEEE Trans. Signal Process.},
  vol.~58, no.~7, pp. 3816--3827, July 2010.

\bibitem{Rabbat1}
M.~Rabbat, J.~D. Haupt, A.~Singh, and R.~D. Nowak, ``Decentralized compression
  and predistribution via randomized gossiping,'' in \emph{Int. Workshop on
  Info. Proc. in Sensor Networks (IPSN)}, Nashville, TN, Apr. 2006.

\bibitem{Patterson1}
S.~Patterson, Y.~C. Eldar, and I.~Keidar, ``Distributed sparse signal recovery
  for sensor networks,'' in \emph{Proc. Acoust., Speech, Signal Processing
  (ICASSP)}, Vancouver, Canada, May 2013.

\bibitem{Baron1}
D.~Baron, M.~Duarte, S.~Sarvotham, M.~B. Wakin, and R.~G. Baraniuk,
  ``Distributed compressed sensing,'' \emph{Rice Univ. Dept. Elect. Comput.
  Eng. Houston, TX, Tech. Rep. TREE–0612}, Nov 2006.

\bibitem{Lopes_2013}
M.~E. Lopes, ``Estimating unknown sparsity in compressed sensing,'' in
  \emph{30$^(th)$ Int. Conf. om Machine Learning (\textsc{JMLR}: \textsc{W}
  $\&$ \textsc{CP})}, Atlanta, GA, 2013, pp. 217--225.

\bibitem{Bioglio_2015}
V.~Bioglio, T.~Bianchi, and E.~Magli, ``On the fly estimation of the sparsity
  degree in compressed sensing using sparse sensing matrices,'' in \emph{IEEE
  Int. Conf. on Acoustics, Speech and Signal Processing (ICASSP)}, South
  Brisbane, QLD, 2015, pp. 3801--3805.

\bibitem{Scharf_TSP94}
L.~L. Scharf and B.~Friedlander, ``Matched subspace detectors,'' \emph{IEEE
  Tran. Signal Process.}, vol.~42, no.~8, pp. 2146 -- 2157, Aug. 1994.

\bibitem{Scharf_ASAP03}
L.~T. McWhorter and L.~L. Scharf, ``Matched subspace detectors for stochastic
  signals,'' in \emph{11th Ann. Workshop on Adaptive Sensor Array Process.
  (ASAP)}, Lexington, MA, Mar. 2003.

\bibitem{jin1}
Y.~Jin and B.~Friedlander, ``A \textsc{CFAR} adaptive subspace detector for
  second-order gaussian signal,'' \emph{IEEE Tran. Signal Process.}, vol.~53,
  no.~3, pp. 871 -- 884, Mar. 2005.

\bibitem{poor1}
H.~V. Poor, \emph{An Introduction to Signal Detection and Estimation}.\hskip
  1em plus 0.5em minus 0.4em\relax New York: Springer-Verlag, 1994.

\bibitem{Sankaran_Bio63}
M.~Sankaran, ``Approximations to the non-central chi-square distribution,''
  \emph{Biometrika}, vol.~50, no. 1-2, pp. 199--204, Dec. 1963.

\bibitem{Junger_2010}
M.~Junger, T.~M. Liebling, D.~Naddef, G.~L. Nemhausera, W.~R. Pulleyblank,
  G.~Reinelt, G.~Rinaldi, and L.~A. Wolsey, \emph{50 Years of Integer
  Programming 1958-2008}.\hskip 1em plus 0.5em minus 0.4em\relax
  Springer-Verlag Berlin Heidelber, 2010.

\bibitem{Leyffer_PhD1993}
S.~Leyffer, ``Deterministic methods for mixed integer nonlinear programming,''
  \emph{Ph.D. dissertation. University of Dundee. Scotland, U.K.}, 1993.

\bibitem{tropp_OMP2007}
J.~Tropp and A.~Gilbert, ``Signal recovery from random measurements via
  orthogonal matching pursuit,'' \emph{IEEE Trans. Inf. Theory}, vol.~53,
  no.~12, pp. 4655--4666, Dec. 2007.

\end{thebibliography}

\end{document}